\documentclass[aps,secnumarabic,nobalancelastpage,amsmath,amssymb,nofootinbib]{revtex4}

\usepackage{hyperref,graphicx}
\usepackage{dcolumn}
\usepackage{bm}
\usepackage{bbm}
\usepackage{tikz}
\usepackage{physics}
\usepackage{enumitem}
\usepackage{amssymb,amsfonts}
\usepackage{epsfig}
\usepackage{epstopdf}
\usepackage[all]{xy}
\usepackage{amsthm}
\usepackage{dcolumn}
\usepackage{hyperref}
\usepackage{url}
\usepackage{dsfont}
\usepackage{slashed}
\usepackage{mathrsfs}
\usepackage{soul}
\usepackage{float}
\usetikzlibrary{tikzmark,
                matrix,
                arrows.meta,
                decorations.pathreplacing}
\usepackage{amsmath}
\usepackage[T1]{fontenc}

\DeclareMathAlphabet{\mathpzc}{OT1}{pzc}{m}{it}

\theoremstyle{plain}

\theoremstyle{remark}

\theoremstyle{plain}
\usepackage{upgreek}
\usepackage{booktabs}

\usepackage{hyperref}
\hypersetup{
    colorlinks=true,
    linkcolor=blue,
    citecolor=blue,
    urlcolor=blue,
}

\newtheorem{Theorem}{Theorem}

\newtheorem{Corollary}{Corollary}

\newtheorem{Definition and Notation}{Definition and Notation}

\begin{document}

\title{Araki Cohomology of the de Sitter Scalar Discrete Series}

\author{Maryam Bajalan$^{1}$\footnote{maryam.bajalan@math.bas.bg}}

\author{Jean-Pierre Gazeau$^{2,3}$\footnote{gazeau@apc.in2p3.fr, j.gazeau@uwb.edu.pl}}

\author{Hamed Pejhan$^{1}$\footnote{pejhan@math.bas.bg}}

\author{Petr Siegl$^{4}$\footnote{siegl@tugraz.at}}

\affiliation{$^1$Institute of Mathematics and Informatics, Bulgarian Academy of Sciences, Acad. G. Bonchev Str. Bl. 8, 1113, Sofia, Bulgaria\\
$^2$Universit\'e Paris Cit\'e, CNRS, Astroparticule et Cosmologie, F-75013 Paris, France\\$^3$Faculty of Mathematics, University of Bia\l ystok, 15-245 Bia\l ystok, Poland\\ $^4$ Institut f\"{u}r Angewandte Mathematik, Technische Universit\"{a}t Graz, A 8010 Graz, Austria}

\date{\today}

\begin{abstract}
    {Building on the recently established de Sitter (dS)-covariant indecomposable Krein realization of the scalar discrete-series unitary irreducible representations $\Pi_{p,0}$, $p=1,2,\cdots$, of $\mathrm{SO}_0(1,4)$, naturally organized into a Gupta-Bleuler triplet, we uncover its underlying Araki-cohomological structure. We show that the Krein-Gupta-Bleuler construction provides a concrete realization of Araki's standard Gupta-Bleuler triplet, with its non-split indecomposable structure encoded by non-trivial first cohomology classes related by adjunction. We further show that the reducible intermediate sector naturally resolves into its two irreducible components, inducing a corresponding decomposition of the extension classes and yielding a four-layer refinement of Araki's standard triplet. The resulting Gupta-Bleuler quadruplet provides a finer cohomological resolution, distinguishing the physical representation from its complex-conjugate sector.}   
\end{abstract}

\maketitle

\setcounter{equation}{0} 
\section{Introduction}

The representation theory of the dS group $\mathrm{SO}_0(1,4)$ provides the natural framework for classifying elementary systems in dS spacetime (see, e.g., Refs. \cite{Gazeau2024, Gazeau2026}). Within its unitary dual, the discrete series \cite{Dixmier, Takahashi} is distinguished by the emergence of gauge structures that are naturally realized through indecomposable representations and Gupta-Bleuler constructions. Such structures occur for strictly massless representations \cite{GravGB, takVec}, partially massless representations \cite{GravPartially}, and, more recently, the scalar discrete series \cite{Bajalan2026}. From a cohomological viewpoint, indecomposability is not merely a feature of the representation matrix; Araki's theory of representations with invariant inner product \cite{Araki-1985-97, Pierotti-1990-31, Simon-1980} identifies the underlying non-split extension structure with non-trivial cohomology classes. The dS discrete series therefore provides a natural setting in which this abstract cohomological structure can be realized explicitly.

In this paper, we realize this cohomological structure explicitly for the scalar discrete-series representations $\Pi_{p,0}$, $p=1,2,\cdots$, of the dS group, building on their recently established dS-covariant realization \cite{Bajalan2026}. For each $p$, $\Pi_{p,0}$ is embedded in an indecomposable representation on a dS-covariant Krein space $V_{\mathrm{tot}}$, endowed with a dS-invariant Klein-Gordon (KG) form and naturally organized into the Krein-Gupta-Bleuler triplet $V_g \subset V \subset V_{\mathrm{tot}}$, where $V_g$ is the gauge subspace and the physical UIR $\Pi_{p,0}$ is carried by $V^\prime \simeq V/V_g$, with ``$\simeq$'' standing for vector space isomorphism. We show that the indecomposable structure of this realization is governed by non-trivial Araki extension classes.

Our main result reveals a two-level cohomological structure underlying the dS-covariant indecomposable realization of the scalar discrete series. At the first level, we show that the representation $\pi$ of $\mathrm{SO}_0(1,4)$ on $V_{\mathrm{tot}}$ realizes Araki's standard Gupta-Bleuler structure,
\begin{align*}
    \pi
    =
    \pi_1^\sharp
    \longrightarrow
    \pi^{}_2
    \longrightarrow
    \pi^{}_1 \,.
\end{align*}
The three representations entering this chain have distinct representation-theoretic roles. The intermediate representation is reducible and decomposes as
\begin{align*}
    \pi^{}_2
    = 
    \pi^{}_{2^\prime} \oplus \pi_{{\overline{2}^\prime}} \,,
\end{align*}
where $\pi^{}_{2^\prime}$ and $\pi_{{\overline{2}^\prime}}$ are two equivalent realizations of the UIR $\Pi_{p,0}$, carried by $V^\prime \simeq V/V_g$ and its pointwise complex conjugate $\overline{V^\prime}$, respectively. In particular, $V^\prime$, endowed with the positive-definite restriction of the KG form, is identified as the physical Hilbert space. By contrast, $\pi^{}_1$ and $\pi_1^\sharp$ are mutually adjoint-conjugate finite-dimensional representations carried by the gauge sector $V_g$ and the supplementary sector $V_s$, respectively, with $V_s$ completing $V_g$ into the non-degenerate subspace $V_g\dotplus V_s$. Note that ``$\dotplus$'' denotes a direct sum for which KG orthogonality is not implied, versus ``$\oplus$'', which denotes a KG-orthogonal direct sum. Although neither $\pi^{}_1$ nor $\pi_1^\sharp$ is a unitary realization of $\Pi_{p,0}$, each is Weyl equivalent to $\Pi_{p,0}$, in the sense of possessing the same eigenvalues of the dS Casimir operators. Thus, the full indecomposable representation $\pi$ links the physical scalar discrete-series sector and its complex-conjugate realization to the finite-dimensional gauge and supplementary sectors within a single dS-covariant structure. Note that, under the natural identification induced by the quotient map, one has $ V_{\mathrm{tot}}/V \simeq \overline{V^\prime}\oplus V_s$, so that $\overline{V^\prime}$ and $V_s$ may be regarded as complementary subspaces of $V_{\mathrm{tot}}/V$.

The successive non-split extensions in this standard three-layer realization are encoded by the non-trivial first cohomology classes
\begin{align*}
    0\neq c_{12}\in H^1(\pi_2, \pi^{}_1) \,,
    \quad
    0\neq c_{12}^{\ast}\in H^1(\pi_1^\sharp, \pi^{}_2) \,,
\end{align*}
which are related by adjunction and satisfy
\begin{align*}
    c^{}_{12}\times c_{12}^{\ast} \in B_s^2(\pi_1^\sharp ,\pi^{}_1) \,.
\end{align*}
This is precisely the cohomological compatibility condition entering Araki's existence theorem for a Gupta-Bleuler triplet in standard form \cite{Araki-1985-97}. The dS Krein-Gupta-Bleuler realization therefore provides an explicit realization of Araki's abstract construction, with its indecomposable character encoded by non-trivial extension classes.

At the second level, the reducibility of the intermediate representation $\pi_2$ reveals a finer cohomological structure. Resolving $\pi_2$ into the two irreducible realizations of $\Pi_{p,0}$ carried separately by $V^\prime$ and $\overline{V^\prime}$ refines the standard triplet into the four-layer realization
\begin{align*}
    \pi = \vcenter{\hbox{%
    \xymatrix@C=1.5em@R=1em{
    & \pi^{}_{2^\prime} \ar[dr] & \\
    \pi_1^\sharp \ar[ur] \ar[dr] & & \pi^{}_1\\
    & \pi_{{\overline{2}^\prime}} \ar[ur] &
    }%
    }}\,.
\end{align*}
Correspondingly, the extension class connecting the intermediate representation to the gauge sector resolves into the two non-trivial classes
\begin{align*}
    0\neq c_{12^\prime} \in H^1(\pi^{}_{2^\prime}, \pi^{}_1) \,, 
    \quad
    0\neq c_{1{\overline{2}^\prime}} \in H^1(\pi_{{\overline{2}^\prime}}, \pi^{}_1) \,,
\end{align*}
while the adjoint extension class connecting the supplementary sector to the intermediate representation resolves as
\begin{align*}
    0\neq c_{12^\prime}^{\ast} \in H^1(\pi_1^\sharp,\pi^{}_{2^\prime}) \,,
    \quad
    0\neq c_{1{\overline{2}^\prime}}^{\ast} \in H^1(\pi_1^\sharp,\pi_{{\overline{2}^\prime}}) \,.
\end{align*}
The corresponding cohomological compatibility relation becomes
\begin{align*}
    c^{}_{12^\prime}\times c_{12^\prime}^{\ast} 
    +
    c_{1{\overline{2}^\prime}}\times c_{1{\overline{2}^\prime}}^{\ast} 
    \in B_s^2(\pi_1^\sharp, \pi^{}_1) \,.
\end{align*}
We thereby obtain a natural four-layer refinement of Araki's standard Gupta-Bleuler triplet---a Gupta-Bleuler quadruplet---in which the two equivalent realizations of $\Pi_{p,0}$ emerge as distinct irreducible cohomological layers, while the finite-dimensional gauge and supplementary representations remain as the mutually adjoint-conjugate outer layers. The refined structure thus separates the physical realization of $\Pi_{p,0}$ from its complex-conjugate counterpart and makes explicit the extension classes through which each is linked to the gauge and supplementary sectors.

The proof develops this structure from the KG geometry of $V_{\mathrm{tot}}$. The associated Araki filtration and projections determine the induced representations and extension maps; the dS representation law yields their cocycle and adjunction relations; and the absence of invariant complements establishes the non-triviality of the corresponding first cohomology classes. The standard triplet and the refined quadruplet thereby emerge as two complementary descriptions of the same dS-covariant indecomposable representation.

The preceding three- and four-layer structures, together with the associated non-trivial extension classes and cohomological compatibility relations, are summarized in Theorem \ref{main thoerem} and proved in detail thereafter.

The paper is organized as follows. In Sec. \ref{Def-all}, we review Araki's cohomological framework for indecomposable representations, with particular emphasis on invariant filtrations, extension maps, and the cohomological characterization of Gupta-Bleuler triplets in standard form. In Sec. \ref{Sec 3}, we apply this framework to the dS scalar discrete-series representations $\Pi_{p,0}$. After recalling the dS-covariant indecomposable Krein-Gupta-Bleuler construction of Ref. \cite{Bajalan2026}, we establish the corresponding Araki filtration and derive the associated projections, induced representations, and extension classes. This leads first to the standard three-layer Gupta-Bleuler realization and subsequently, by resolving its reducible intermediate sector, to the refined four-layer cohomological structure stated in Theorem \ref{main thoerem}. In Sec. \ref{Sec conclu}, we conclude with a discussion of the representation-theoretic implications of these results and of the directions they suggest for the dS discrete series.

\section{Araki's Cohomological Theory of Indecomposable Representations} \label{Def-all}

Following Araki \cite{Araki-1985-97} (see also \cite{Pierotti-1990-31, Simon-1980}), let $\pi\colon\; G\longrightarrow \mathrm{GL}({\mathfrak{H}})$ be a representation of a Lie group $G$ on a complex vector space ${\mathfrak{H}}$, endowed with a $G$-invariant, non-degenerate Hermitian sesquilinear form $\langle\cdot,\cdot\rangle$. Let
\begin{align}
    (\left\{0\right\} =) \,{\mathfrak{H}}^{}_0 \;\;\subset\;\; {\mathfrak{H}}^{}_1 \;\;\subset\;\; \cdots \;\;\subset\;\; {\mathfrak{H}}_n \,(= {\mathfrak{H}})
\end{align}
be a filtration by closed $G$-invariant subspaces such that, for every $j=1,\cdots,n$, the subspace ${\mathfrak{H}}_{j-1}$ admits no \emph{closed} $G$-invariant complement in ${\mathfrak{H}}_j$. Here, closedness is understood with respect to the Hermitian sesquilinear form $\langle\cdot,\cdot\rangle$; a subspace ${\mathfrak{E}}\subset{\mathfrak{H}}$ is called \emph{closed} if
\begin{align}
    {\mathfrak{E}} = ({\mathfrak{E}}^{\perp})^{\perp} \,, \quad\mbox{where}\quad {\mathfrak{E}}^{\perp} := \left\{ \xi\in{\mathfrak{H}} \,;\, \left\langle\xi,\eta\right\rangle = 0 \;\;\text{for all}\; \eta\in{\mathfrak{E}} \right\} \,.
\end{align}
Each inclusion ${\mathfrak{H}}_{j-1}\subset{\mathfrak{H}}_j$ defines a non-split extension of representations, so that the filtration realizes $\pi$ as a successive sequence of non-split \emph{indecomposable}\footnote{A representation $\pi\colon\; G\longrightarrow \mathrm{GL}({\mathfrak{H}})$ of a Lie group $G$ on a complex vector space ${\mathfrak{H}}$, endowed with a $G$-invariant, non-degenerate Hermitian sesquilinear form $\langle\cdot,\cdot\rangle$, is called \emph{irreducible} if ${\mathfrak{H}}$ contains no non-zero proper closed $G$-invariant subspace. It is called \emph{topologically indecomposable} if there exist no non-zero closed $G$-invariant subspaces ${\mathfrak{E}}_1$ and ${\mathfrak{E}}_2$ satisfying
\begin{align*}
    {\mathfrak{E}}_1 \cap {\mathfrak{E}}_2 = \left\{0\right\} \,, \quad \left(({\mathfrak{E}}_1 + {\mathfrak{E}}_2)^\perp\right)^\perp = {\mathfrak{H}} \,.
\end{align*}
Every irreducible representation is topologically indecomposable; however, the converse does not hold in general.\label{footnote indecom}} extensions.

For each $j=1,\cdots,n$, let ${\mathfrak{R}}^{}_j$ be a complement of ${\mathfrak{H}}_{j-1}$ in ${\mathfrak{H}}_j$, so that
\begin{align} \label{R1+...+Rn''}
    {\mathfrak{H}}_j = {\mathfrak{H}}_{j-1} \oplus {\mathfrak{R}}^{}_j \,, \quad\mbox{where}\quad {\mathfrak{R}}^{}_j \simeq {\mathfrak{H}}_j/{\mathfrak{H}}_{j-1} \,.
\end{align}
Note that, throughout the text, the symbol ``$\simeq$'' denotes an isomorphism of vector spaces or representations, with its precise meaning determined by the context. Accordingly, one obtains the decomposition
\begin{align} \label{R1+...+Rn}
    {\mathfrak{H}} = {\mathfrak{R}}^{}_1 \oplus \cdots \oplus {\mathfrak{R}}^{}_n \,, \quad\mbox{with}\quad \mathfrak{R}_1=\mathfrak{H}^{}_1 \,.
\end{align} 
The quotient space ${\mathfrak{H}}_j/{\mathfrak{H}}_{j-1} \simeq {\mathfrak{R}}^{}_j$ carries the induced representation
\begin{align}
    \pi_j(\texttt{g})[\xi] = [\pi(\texttt{g})\xi] \,,
\end{align}
where $\texttt{g}\in G$ and $[\xi] = \xi + {\mathfrak{H}}_{j-1}$ denotes the equivalence class of $\xi\in{\mathfrak{H}}_j$ in ${\mathfrak{H}}_j/{\mathfrak{H}}_{j-1}$. A representation \linebreak $\pi_j^\sharp \,\colon\; G \longrightarrow \mathrm{GL}(\mathfrak{R}_j^\sharp)$ is said to be \emph{adjoint conjugate} to a representation $\pi_j \colon\; G \longrightarrow \mathrm{GL}(\mathfrak{R}_j)$ if there exists a sesquilinear pairing
\begin{align}
    \left\langle \cdot,\cdot\right\rangle_\sharp \colon\; \mathfrak{R}_j^\sharp\times \mathfrak{R}_j^{} \longrightarrow \mathbb{C} \,,
\end{align}
such that:
\begin{enumerate}
    \item{$\mathfrak{R}_j^\sharp$ and $\mathfrak{R}^{}_j$ separate each other, namely,
    \begin{align}\label{3333'}
        \left\langle \phi , \psi \right\rangle_\sharp = 0 \quad\text{for all}\quad \psi \in \mathfrak{R}_j \quad\Longrightarrow\quad \phi = 0 \,, \nonumber\\[3pt]
        \left\langle \phi , \psi \right\rangle_\sharp = 0 \quad\text{for all}\quad \phi \in \mathfrak{R}_j^\sharp\quad\Longrightarrow\quad \psi = 0 \,.
    \end{align}}

    \item{The pairing is $G$-invariant in the sense that, for all $\phi \in {\mathfrak{R}}_j^\sharp$ and $\psi \in {\mathfrak{R}}^{}_j$,
    \begin{align}\label{conj-def}
        \left\langle \pi_j^\sharp (\texttt{g})\phi \,,\, \pi^{}_j(\texttt{g}) \psi \right\rangle_{\sharp} = \left\langle \phi \,,\, \psi \right\rangle_{\sharp} \,.
    \end{align}}
\end{enumerate}

The pairing $\langle\cdot,\cdot\rangle_{\sharp}$, introduced in the definition of adjoint-conjugate representations, is, in the present setting where all representations are realized within a single representation space $({\mathfrak{H}}, \langle\cdot,\cdot\rangle)$, naturally induced by the original $G$-invariant Hermitian form. More precisely, whenever ${\mathfrak{R}}^{}_i \simeq {\mathfrak{R}}_j^\sharp$ and $\pi^{}_i \simeq \pi_j^\sharp$, with the restriction of $\langle\cdot,\cdot\rangle$ to ${\mathfrak{R}}^{}_i \times {\mathfrak{R}}^{}_j$ non-degenerate, one may identify
\begin{align}
    \langle\phi,\psi\rangle_{\sharp} := \langle\phi,\psi\rangle \,, \quad \phi\in{\mathfrak{R}}^{}_i\,,\quad \psi\in{\mathfrak{R}}^{}_j \,.
\end{align}
Henceforth, we shall omit the subscript ``$\,\sharp\,$'' and simply write $\langle\cdot,\cdot\rangle$ for the induced pairing between adjoint-conjugate representation spaces.

Let $\pi_j$ and $\pi_i$ be representations of $G$ on ${\mathfrak{R}}^{}_j$ and ${\mathfrak{R}}^{}_i$, respectively. We adopt the following standard cohomological notation and terminology:
\begin{enumerate}
    \item{$C^n(\pi_j,\pi_i)$ denotes the vector space of all \emph{$n$-cochains} $c_{ij}$ that assign to each $n$-tuple $(\texttt{g}_1, \cdots, \texttt{g}_n)\in G^n$ a linear operator $c_{ij}(\texttt{g}_1, \cdots, \texttt{g}_n) \colon\; {\mathfrak{R}}^{}_j \longrightarrow  {\mathfrak{R}}^{}_i$. For $n=0$, $c_{ij}\in C^0(\pi_j,\pi_i)$ is a fixed linear map from $\mathfrak{R}_j$ to $\mathfrak{R}_i$.}
    
    \item{The product of an $n$-cochain $c_{jk}\in C^n(\pi_k,\pi_j)$ and an $m$-cochain $c_{ki}\in C^m(\pi_i,\pi_k)$ is defined by
    \begin{align}\label{pro-eq}
        (c_{jk} \times c_{ki})(\texttt{g}_1, \cdots, \texttt{g}_{n+m}) := c_{jk}(\texttt{g}_1, \cdots, \texttt{g}_n) \,c_{ki}(\texttt{g}_{n+1}, \cdots, \texttt{g}_{n+m}) \,,
    \end{align}
    where juxtaposition on the right-hand side denotes composition of linear maps.}
    
    \item{The \emph{$n$-coboundary} operator $\delta \colon\; C^n(\pi_j,\, \pi_i)\longrightarrow C^{n+1}(\pi_j, \, \pi_i)$ is defined, for any $c_{ij} \in C^n(\pi_j,\pi_i)$, by
    \begin{align}\label{sum-cn}
        (\delta c_{ij})(\texttt{g}_1,\cdots,\texttt{g}_{n+1}) =&\, \pi_i(\texttt{g}_1) \,c_{ij}(\texttt{g}_2,\cdots,\texttt{g}_{n+1}) \nonumber\\[3pt]
        &+ \sum_{k=1}^{n} (-1)^k \,c_{ij} (\texttt{g}_1,\cdots, \texttt{g}_k\texttt{g}_{k+1}, \cdots, \texttt{g}_{n+1}) \nonumber\\[3pt]
        &+ (-1)^{n+1} \,c_{ij}(\texttt{g}_1,\cdots,\texttt{g}_n) \,\pi_j(\texttt{g}_{n+1}) \,, \hspace{3cm} (\mbox{satisfying}\;\;\delta^2 = 0) \,.
    \end{align}
    Specifically, for $n=0$, Eq. \eqref{sum-cn} reduces to
    \begin{align}\label{sum-cn00}
        (\delta c_{ij})(\texttt{g}) = \pi_i(\texttt{g})c_{ij} - c_{ij}\pi_j(\texttt{g}) \,.
    \end{align}}

    \item{The spaces of \emph{$n$-cocycles}, \emph{$n$-coboundaries}, and the \emph{$n$-cohomology group} are defined as
    \begin{align}
        \label{Def-1-co} Z^n(\pi_j,\pi_i) &:= \left\{c_{ij}\in C^n(\pi_j,\pi_i) \;;\; \delta c_{ij} = 0 \right\} \,, \\[3pt]
        \label{Def-2-co} B^n(\pi_j,\pi_i) &:= \delta C^{n-1}(\pi_j,\pi_i) \,, \quad n\geqslant 1 \,, \qquad (B^0(\pi_j,\pi_i):=\left\{0\right\}) \,, \\[3pt]
        \label{Def-3-co} H^n(\pi_j,\pi_i) &:= Z^n(\pi_j,\pi_i) \big/ B^n(\pi_j,\pi_i) \,.
    \end{align}
    The elements of $H^n(\pi_j,\pi_i)$ are equivalence classes $[c_{ij}]$ of $n$-cocycles modulo $n$-coboundaries. Following Araki \cite{Araki-1985-97}, throughout we use the same symbol $c_{ij}$ for a cocycle and the cohomology class it represents; thus, an expression such as $0\neq c_{ij}\in H^1(\pi_j,\pi_i)$ is understood as $0\neq[c_{ij}]\in H^1(\pi_j,\pi_i)$, with $c_{ij}\in Z^1(\pi_j,\pi_i)$ a representative of the class.}

    \item{The space $C_b^n(\pi_j,\pi_i)$ consists of all $n$-cochains $c_{ij}\in C^n(\pi_j,\pi_i)$ for which there exists an $n$-cochain $c_{ij}^\ast\in C^n(\pi_i^\sharp,\pi_j^\sharp)$ satisfying
    \begin{align}\label{eq:3.1}
        \left\langle \phi \,,\, c^{}_{ij}(\texttt{g}_1,\cdots\,,\texttt{g}_n)\psi \right\rangle = \left\langle c^\ast_{ij} (\texttt{g}_n^{-1},\cdots, \,\texttt{g}_1^{-1})\phi \,,\, \psi \right\rangle \,,
    \end{align}
    for every $\psi\in{\mathfrak{R}}^{}_j$, $\phi\in{\mathfrak{R}}_i^\sharp$, and $(\texttt{g}_1, \cdots, \texttt{g}_n)\in G^n$. The following properties then follow immediately from Eq. \eqref{eq:3.1},
    \begin{align}\label{khodshifte}
        \left(c^{\ast}_{jk}\right)^\ast = c^{}_{jk} \,, \quad\mbox{and}\quad \left(c^{}_{jk} c^{}_{ki}\right)^\ast = c^{\ast}_{ki} \,c^{\ast}_{jk} \,.
    \end{align}
    Moreover, for $c_{ij}\in C_b^n(\pi_j,\pi_i)$, the adjoint operation and the coboundary operator are related by
    \begin{align}\label{maryam}
        \left(\delta c_{ij}\right)^\ast
        =(-1)^{n+1}\, \delta\left(c_{ij}^\ast\right) \,.
    \end{align}}

    \item{A cochain $c_{ij}$ is called self-adjoint if $c_{ij}^\ast = c^{}_{ij}$. The space of all self-adjoint cochains is denoted by $C_s^n(\pi_j,\pi_i)$. The corresponding space of self-adjoint coboundaries is defined by $B_s^n := \delta C_s^{n-1}$.}
\end{enumerate}

The preceding constructions culminate in the operator-valued upper-triangular realization of the representation $\pi$. More precisely, relative to the decomposition \eqref{R1+...+Rn}, one has
\begin{align}\label{matrix}
    \pi(\texttt{g}) =
    \begin{pmatrix}
        \pi^{}_1(\texttt{g}) & c_{12}(\texttt{g}) & \cdots & c_{1n}(\texttt{g}) \\
        0 & \pi^{}_2(\texttt{g}) & \cdots & c_{2n}(\texttt{g}) \\
        \vdots & \ddots & \ddots & \vdots \\
        0 & \cdots & 0 & \pi_n(\texttt{g})
    \end{pmatrix} \,,
\end{align}
or, schematically,
\begin{align}
    \pi = \pi_n \longrightarrow \cdots \longrightarrow \pi^{}_1 \,,
\end{align}
where the operator-valued $1$-cochains $c_{ij}\in C^1(\pi_j,\pi_i)$ encode the extension data.

Araki proved that the absence of $G$-invariant complements admits the following cohomological characterization.
\begin{Corollary}\label{Prop}\cite[Cor. 4.4]{Araki-1985-97} 
    For a representation $\pi$ of the form \eqref{matrix}, $\mathfrak{H}_j$ has no $G$-invariant complement in $\mathfrak{H}_{j+1}$ for each $j = 1, \cdots, n-1$ if and only if the following conditions hold
    \begin{enumerate}
        \item{
        \begin{align}
            c_{k, k+1} \notin B^1(\pi_{k+1}, \pi_k)\,, \quad\mbox{for}\quad k = 1, \cdots, n-1 \,.
        \end{align}}

        \item{
        \begin{align}
            \delta c_{kl} = -\sum_{j=k+1}^{l-1} c_{kj} \times c_{jl}\,, \quad\mbox{for}\quad 1 \leqslant k < l +1 \leqslant n+1 \,.
        \end{align}}        
    \end{enumerate}
\end{Corollary}

The representation $\pi$ in Eq. \eqref{matrix} is said to define a \emph{Gupta-Bleuler triplet in the standard form}, if the following conditions are satisfied:
\begin{enumerate}
    \item{$n=3$, namely, 
    \begin{align}
        \left\{0\right\} \;\;\subset\;\; {\mathfrak{H}}^{}_1 \left(= {\mathfrak{R}}^{}_1\right) \;\;\subset\;\; {\mathfrak{H}}^{}_2 \;\;\subset\;\; {\mathfrak{H}}^{}_3 \left(= {\mathfrak{R}}^{}_1 \oplus {\mathfrak{R}}^{}_2 \oplus {\mathfrak{R}}^{}_3 \right) \,.
    \end{align}}

    \item{
    \begin{align}\label{2=1perp}
        {\mathfrak{H}}^{}_2 = {{\mathfrak{H}}^{}_1}^{\perp} \,.
    \end{align}}

    \item{
    \begin{align}\label{33333}
        {\mathfrak{R}}_3^{} \simeq {\mathfrak{R}}_1^\sharp\,, \quad\mbox{and}\quad \pi_3^{} \simeq \pi_1^\sharp \,.
    \end{align}}
\end{enumerate}
In this case, the representation $\pi$ takes the form
\begin{align}\label{standard}
    \pi(\texttt{g}) &=
    \begin{pmatrix}
        \pi^{}_1(\texttt{g}) & c_{12}(\texttt{g}) & c_{13}(\texttt{g})\\
        0 & \pi^{}_2(\texttt{g}) & c_{23}(\texttt{g})\\
        0 & 0 & \pi_1^\sharp(\texttt{g})
    \end{pmatrix} \,,\\[3pt]
     \pi &= \pi_1^\sharp \longrightarrow \pi^{}_2 \longrightarrow \pi^{}_1 \,,
\end{align}
while the ambient $G$-invariant Hermitian form is given by
\begin{align}
    \left\langle\phi_1+\phi_2+\phi_3 \,,\, \psi_1+\psi_2+\psi_3\right\rangle = \overline{\left\langle \psi_3 \,,\, \phi_1 \right\rangle} + \left\langle \phi_3 \,,\, \psi_1 \right\rangle + \left\langle\phi_2 \,,\, \psi_2\right\rangle \,,
\end{align}
for $\phi_j,\psi_j\in{\mathfrak{R}}^{}_j$, with $j=1,2,3$.

The necessary and sufficient cohomological conditions for the existence of a Gupta-Bleuler triplet in the standard form are provided by Araki's existence theorem.
\begin{Theorem}\label{Theorem1}\cite[Thm. 2]{Araki-1985-97}
    Let $\pi^{}_1$ and $\pi_1^\sharp$ be a pair of adjoint-conjugate irreducible representations, and let $\pi_2$ be a representation admitting a $G$-invariant non-degenerate Hermitian form. Then, a Gupta-Bleuler triplet in the standard form exists if and only if there are non-zero classes $c^{}_{12}\in H^1(\pi^{}_2,\pi^{}_1)$ and $c_{12}^{\ast}\in H^1(\pi_1^\sharp,\pi^{}_2)$ such that
    \begin{align}
        c^{}_{12}\times c_{12}^\ast \,\in\, B_s^2(\pi_1^\sharp,\pi^{}_1) \,.
    \end{align}    
\end{Theorem}

Thus, under the hypotheses of the theorem, the existence of a Gupta-Bleuler triplet is completely characterized by two non-trivial first cohomology classes whose product is a self-adjoint coboundary.

\section{Araki Cohomology of the dS Scalar Discrete Series $\Pi_{p,0}$}\label{Sec 3}

Having established the general cohomological framework, we now apply it to the scalar discrete-series realization. After briefly recalling the dS-covariant indecomposable Krein-Gupta-Bleuler construction of Ref. \cite{Bajalan2026}, we identify its associated Araki filtration and extension classes, first in the standard Gupta-Bleuler triplet and then in its refined four-layer realization.

\begin{table*}[t]
    \centering
    \renewcommand{\arraystretch}{1.4}
    \setlength{\tabcolsep}{16pt}
    \begin{tabular}{llll}
    \hline\hline
    \noalign{\vskip 3pt}
    \text{Space} & \text{dS-invariant hull} & \text{KG orthogonality} & \text{KG pairing} \\[3pt]
    \hline\hline \noalign{\vskip 3pt}
    $V_g = \overline{V_g}$ & $V_g$ & ${V_g}^{\perp_{\vphantom{}_{\mathrm{KG}}}} = V^\prime \oplus \overline{V^\prime} \oplus V_g$ & $\langle \phi^{(g)}_{L<p} \,,\, \phi^{(s)}_{L^\prime <p} \rangle_{\vphantom{}_{\mathrm{KG}}} = \delta_{LL^\prime}$ \\[3pt]
    $= \operatorname{Span}_{\mathbb{C}}\left\{\phi^{(g)} := \phi^{(g)}_{L<p}\right\}$&  & $({V_g}^{\perp_{\vphantom{}_{\mathrm{KG}}}})^{\perp_{\vphantom{}_{\mathrm{KG}}}} = V_g$ & $\langle \phi^{(g)}_{L<p} \,,\, \phi^{(g)}_{L^\prime <p} \rangle_{\vphantom{}_{\mathrm{KG}}} = 0$ \\[3pt]
    \hline \noalign{\vskip 3pt}
    $V^\prime \, \left(=: V_{\mathrm{phys}}\right)$ & $V^\prime \oplus V_g$ & ${V^\prime}^{\perp_{\vphantom{}_{\mathrm{KG}}}} = \overline{V^\prime} \oplus V_g \oplus V_s$ & $\langle \phi^{}_{L\geqslant p} \,,\, \phi^{}_{L^\prime\geqslant p} \rangle_{\vphantom{}_{\mathrm{KG}}} = \delta_{LL^\prime}$ \\[3pt]
    $= \operatorname{cl}_{\mathrm{KG}} \operatorname{Span}_{\mathbb{C}}\left\{\phi^{} := \phi^{}_{L\geqslant p}\right\}$&  &  $({V^\prime}^{\perp_{\vphantom{}_{\mathrm{KG}}}})^{\perp_{\vphantom{}_{\mathrm{KG}}}} = V^\prime$ & $\langle \phi^{}_{L\geqslant p} \,,\, \overline{\phi^{}_{L^\prime\geqslant p}} \rangle_{\vphantom{}_{\mathrm{KG}}} = 0$ \\[3pt]
    \hline \noalign{\vskip 3pt}
    $V = V^\prime\oplus V_g$ & $V$ & & $\langle \phi^{}_{L\geqslant p} \,,\, \phi^{(g)}_{L^\prime < p} \rangle_{\vphantom{}_{\mathrm{KG}}} = 0$ \\[3pt]
    & & &  \\[3pt]
    \hline \noalign{\vskip 3pt}
    $\overline{V^\prime}$ & $\overline{V^\prime}\oplus \overline{V_g}$ & $\overline{V^\prime}^{\perp_{\vphantom{}_{\mathrm{KG}}}} = V^\prime \oplus V_g \oplus V_s$ & $\langle \overline{\phi^{}_{L\geqslant p}} \,,\, -\overline{\phi^{}_{L^\prime\geqslant p}} \rangle_{\vphantom{}_{\mathrm{KG}}} = \delta_{LL^\prime}$ \\[3pt]
    $= \operatorname{cl}_{\mathrm{KG}} \operatorname{Span}_{\mathbb{C}}\left\{ \overline{\phi^{}} := \overline{\phi^{}_{L\geqslant p}}\right\}$ &  & $(\overline{V^\prime}^{\perp_{\vphantom{}_{\mathrm{KG}}}})^{\perp_{\vphantom{}_{\mathrm{KG}}}} = \overline{V^\prime}$ & $\langle \overline{\phi^{}_{L\geqslant p}} \,,\, \phi^{}_{L^\prime \geqslant p} \rangle_{\vphantom{}_{\mathrm{KG}}} = 0$ \\[3pt]
    \hline\noalign{\vskip 3pt}
    $V_s = \overline{V_s}$ & $V_s \oplus \overline{V^\prime} \oplus V^\prime \oplus V_g$ & ${V_s}^{\perp_{\vphantom{}_{\mathrm{KG}}}} = V^\prime \oplus \overline{V^\prime} \oplus V_s$ & $\langle \phi^{(s)}_{L<p} \,,\, \phi^{(g)}_{L^\prime<p} \rangle_{\vphantom{}_{\mathrm{KG}}} = \delta_{LL^\prime}$\\[3pt]
    $= \operatorname{Span}_{\mathbb{C}}\left\{\phi^{(s)} := \phi^{(s)}_{L<p}\right\}$&  & $({V_s}^{\perp_{\vphantom{}_{\mathrm{KG}}}})^{\perp_{\vphantom{}_{\mathrm{KG}}}} = V_s$ & $\langle \phi^{(s)}_{L< p} \,,\, \phi^{(s)}_{L^\prime < p} \rangle_{\vphantom{}_{\mathrm{KG}}} = 0$ \\[3pt]
    \hline \noalign{\vskip 3pt}
    $H_+$ & $H_+ \oplus \overline{V^\prime}$ &  & $\langle \phi^{}_{L\geqslant 0} \,,\, \phi^{}_{L^\prime \geqslant 0} \rangle_{\vphantom{}_{\mathrm{KG}}} = \delta_{LL^\prime}$ \\[3pt]
    $= \operatorname{cl}_{\mathrm{KG}} \operatorname{Span}_{\mathbb{C}}\left\{\phi^{}_{L<p} \,,\; \phi^{}_{L\geqslant p}\right\}$ & & & \\[3pt]
    $\phi^{}_{L<p} := \phi^{(g)}_{L<p} + \tfrac{1}{2} \phi^{(s)}_{L<p}$ & & & \\[3pt]
    \hline \noalign{\vskip 3pt}
    $H_-$ &  $H_- \oplus V^\prime$ &  & $\langle \overline{\phi^{}_{L\geqslant 0}} \,,\, - \overline{\phi^{}_{L^\prime \geqslant 0}} \rangle_{\vphantom{}_{\mathrm{KG}}} = \delta_{LL^\prime}$ \\[3pt]
    $= \operatorname{cl}_{\mathrm{KG}} \operatorname{Span}_{\mathbb{C}}\left\{\overline{\phi^{}_{L<p}} \,,\; \overline{\phi^{}_{L\geqslant p}}\right\}$ & & & \\[3pt]
    $\overline{\phi^{}_{L<p}} := \phi^{(g)}_{L<p} - \tfrac{1}{2} \phi^{(s)}_{L<p}$ & & & \\[3pt]
    \hline \noalign{\vskip 3pt}
    $V_{\mathrm{tot}} = H_+\oplus H_-$ & $V_{\mathrm{tot}}$ & & $\langle \phi^{}_{L\geqslant 0} \,,\,  \overline{\phi^{}_{L^\prime \geqslant 0}} \rangle_{\vphantom{}_{\mathrm{KG}}} = 0$ \\[3pt]
    $= \operatorname{cl}_{\mathrm{KG}} \operatorname{Span}_{\mathbb{C}} \left\{\Phi \right\}$ & & & \\[3pt]
    $\left\{\Phi\} := \{\phi^{}_{L\geqslant 0} \,,\; \overline{\phi^{}_{L\geqslant 0}}\right\}$ &  &  & \\[3pt]
    \hline\hline \noalign{\vskip 3pt}
    \end{tabular}
    \caption{Structure of the subspaces entering the dS-covariant indecomposable Krein-Gupta-Bleuler realization, together with their dS-invariant hulls and the relevant KG-orthogonality and pairing relations. Throughout, $\operatorname{Span}_{\mathbb{C}}\{a\}$ denotes the algebraic complex linear span of a family $\{a\}$, whereas $\operatorname{cl}_{\mathrm{KG}}\operatorname{Span}_{\mathbb{C}}\{a\}$ denotes its completion with respect to the corresponding Hilbert norm induced by the KG form. Accordingly, the finite-dimensional spaces $V_g$ and $V_s$ are already complete and are given simply by their algebraic spans, while $V^\prime$, $\overline{V^\prime}$, $H_+$, $H_-$, and $V_{\mathrm{tot}}$ are the corresponding KG-completed spaces. The overline denotes pointwise complex conjugation. The bases $\{\phi^{(g)}\}$ and $\{\phi^{(s)}\}$ of $V_g$ and $V_s$, respectively, are KG-paired, while $\{\phi\}$ and $\{\overline{\phi}\}$ are KG-``orthonormal'' bases of $V^\prime$ and its pointwise complex conjugate $\overline{V^\prime}$, respectively. \textbf{Entries that follow directly from those displayed are omitted.}}
    \label{Table 1}
\end{table*}

\subsection{Construction Recalled from Ref. \cite{Bajalan2026}}

For each $p=1,2,\cdots$, the scalar discrete-series UIR $\Pi_{p,0}$ of the dS group $G = \mathrm{SO}_0(1,4)$ is realized dS-covariantly within an indecomposable representation on a Krein space $V_{\mathrm{tot}}$, endowed with the dS-invariant, non-degenerate Hermitian Klein-Gordon (KG) sesquilinear form \cite{Birrell, Higuchi},
\begin{align}
    \left\langle \phi^{}_1(x) \,,\, \phi^{}_2(x) \right\rangle_{\vphantom{}_{\mathrm{KG}}} = \mathrm{i}\int_{\Sigma} \left(\overline{\phi^{}_1(x)}\, {\partial}_{s}\, \phi^{}_2(x) - \phi^{}_2(x)\, {\partial}_{s}\, \overline{\phi^{}_1(x)} \right) \mathrm{d}\sigma^s \,,
\end{align}
where $\Sigma$ is a Cauchy surface and $\mathrm{d}\sigma^s$ is the corresponding directed surface element. Throughout, an overline denotes pointwise complex conjugation, both for individual vectors and for subspaces.

The Krein space $V_{\mathrm{tot}}$ admits the decomposition
\begin{align}
    V_{\mathrm{tot}} = (V_g\dotplus V_s) \oplus V^\prime \oplus \overline{V^\prime} \,.
\end{align}
Here, $V_g = \overline{V_g}$ denotes the finite-dimensional gauge-state space ($\dim V_g=p(p+1)(2p+1)/6$), while $V_s=\overline{V_s}$ is the supplementary-state space ($\dim V_s=\dim V_g$). Both spaces are degenerate with respect to the KG form. Specifically, let $\left\{\phi^{(g)}_{L<p}\right\}$ and $\left\{\phi^{(s)}_{L<p}\right\}$ denote the bases of $V_g$ and $V_s$, respectively.\footnote{The angular-momentum indices range over $\left\{(L,l,m)\in\mathbb{N}_0\times\mathbb{N}_0\times\mathbb{Z} \;;\; 0\leqslant l\leqslant L \,,\; |m|\leqslant l\right\}$, where $\mathbb{N}_0 = \{0,1,2,\cdots\}$. Since the admissible values of $l$ and $m$ are determined by $L$, we suppress these indices throughout and retain only the principal angular-momentum index $L$.} Their KG norms vanish, 
\begin{align}\label{gs}
    \left\langle \phi^{(g)}_{L<p} \,,\, \phi^{(g)}_{L^\prime <p} \right\rangle_{\vphantom{}_{\mathrm{KG}}} 
    = 0 = 
    \left\langle \phi^{(s)}_{L<p} \,,\, \phi^{(s)}_{L^\prime <p} \right\rangle_{\vphantom{}_{\mathrm{KG}}} \,.
\end{align} 
The two bases are, however, mutually KG-paired,
\begin{align}\label{gs}
    \left\langle \phi^{(g)}_{L<p} \,,\, \phi^{(s)}_{L^\prime <p} \right\rangle_{\vphantom{}_{\mathrm{KG}}} = \delta_{LL^\prime} = \left\langle \phi^{(s)}_{L<p} \,,\, \phi^{(g)}_{L^\prime <p} \right\rangle_{\vphantom{}_{\mathrm{KG}}} \,.
\end{align} 
Indeed, the finite-dimensional spaces $V_g$ and $V_s$ carry mutually adjoint-conjugate finite-dimensional representations, each Weyl equivalent to the UIR $\Pi_{p,0}$, in the sense of possessing the same eigenvalues of the dS Casimir operators as $\Pi_{p,0}$.

It then follows that the space
\begin{align}
    V_g\dotplus V_s := \operatorname{Span}_{\mathbb{C}}\left\{\phi^{}_{L<p} := \phi^{(g)}_{L<p} + \tfrac{1}{2} \phi^{(s)}_{L<p} 
    \quad,\quad
    \overline{\phi^{}_{L<p}} := \phi_{L<p}^{(g)} - \tfrac{1}{2}\phi_{L<p}^{(s)}\right\} 
\end{align}
is non-degenerate with respect to the KG form. Indeed,
\begin{align}
    \left\langle \phi^{}_{L<p} \,,\, \phi^{}_{L^\prime<p} \right\rangle_{\vphantom{}_{\mathrm{KG}}} = \delta_{LL^\prime} = \left\langle \overline{\phi^{}_{L<p}} \,,\, - \overline{\phi^{}_{L^\prime<p}} \right\rangle_{\vphantom{}_{\mathrm{KG}}} \,, 
    \quad\mbox{and}\quad 
    \left\langle \phi^{}_{L<p} \,,\, \overline{\phi^{}_{L^\prime<p}} \right\rangle_{\vphantom{}_{\mathrm{KG}}} = 0 \,.
\end{align} 
Note that, throughout what follows, ``$\oplus$'' denotes a KG-orthogonal direct sum, whereas ``$\dotplus$'' denotes a direct sum for which KG orthogonality is not implied.

The subspace $V^\prime$ and its pointwise complex conjugate $\overline{V^\prime}$ are the infinite-dimensional KG-positive and KG-negative sectors, respectively, with bases $\left\{\phi^{}_{L\geqslant p}\right\}$ and $\left\{\overline{\phi^{}_{L\geqslant p}}\right\}$ normalized according to
\begin{align}
    \left\langle \phi^{}_{L\geqslant p} \,,\, \phi^{}_{L^\prime\geqslant p} \right\rangle_{\vphantom{}_{\mathrm{KG}}} = \delta_{LL^\prime} = \left\langle \overline{\phi^{}_{L\geqslant p}} \,,\, -\overline{\phi^{}_{L^\prime\geqslant p}} \right\rangle_{\vphantom{}_{\mathrm{KG}}} \,, 
    \quad\mbox{and}\quad 
    \left\langle \phi^{}_{L\geqslant p} \,,\, \overline{\phi^{}_{L^\prime\geqslant p}} \right\rangle_{\vphantom{}_{\mathrm{KG}}} = 0 \,.
\end{align}
The spaces $V^\prime$ and $\overline{V^\prime}$ carry two equivalent realizations of the UIR $\Pi_{p,0}$, with $V^\prime$, endowed with the positive-definite restriction of the KG form, identified as the physical Hilbert space.

The resulting structure gives rise to the $G=\mathrm{SO}_0(1,4)$-invariant Gupta-Bleuler triplet associated with the indecomposable realization of $\Pi_{p,0}$,
\begin{align}\label{KGB}
    V_g \;\;\subset\;\; V \,\left(:=V_g\oplus V^\prime\right) \;\;\subset\;\; V_{\mathrm{tot}} \,.
\end{align}
We denote by $\left\{\Phi \right\} = \left\{ \phi^{}_{L<p}\,,\; \overline{\phi^{}_{L<p}} \,,\; \phi^{}_{L\geqslant p} \,,\; \overline{\phi^{}_{L\geqslant p}}\right\}$ the KG-orthonormal basis of $V_{\mathrm{tot}}$, with the positive- and negative-norm sectors normalized as above. An explicit construction of these basis elements is given in Ref. \cite{Bajalan2026}. The relevant subspaces, their dS-invariant hulls, and the corresponding KG-orthogonality and pairing relations are summarized in Table \ref{Table 1}.

\subsection{Araki Cohomological Formulation}

In this subsection, we show that the above Krein-Gupta-Bleuler structure admits a natural formulation within the cohomology of group representations in the sense of Araki \cite{Araki-1985-97, Pierotti-1990-31, Simon-1980}.

The Gupta-Bleuler chain \eqref{KGB} is the natural filtration for the construction of the physical quotient $V_{\mathrm{phys}} = V^\prime \simeq V/V_g$. For the comparison with Araki's standard form, however, one must use the filtration determined by the KG-orthogonal complement of the first invariant subspace (see Eq. \eqref{33333}). Accordingly, in view of the invariant subspace structure summarized in Table \ref{Table 1}, we introduce the corresponding Araki filtration (see Fig. \ref{Fig})
\begin{align} \label{32}
    \underbrace{\quad\mathfrak{H}^{}_0\quad}_{:=\, \left\{0\right\}}
    \quad\subset\quad \underbrace{\quad\mathfrak{H}^{}_1\quad}_{:=\, V_g} \quad\left( \subset\quad \underbrace{\quad\mathfrak{H}^{}_{2^\prime}\quad}_{:=\, V_g \oplus V^\prime} \; \right) \quad\subset\quad \underbrace{\quad\mathfrak{H}^{}_2 \; \,\left(= {\mathfrak{H}}_{{\overline{2}^\prime}}\right) \quad}_{:=\, V_g \oplus V^\prime \oplus \overline{V^\prime}} \quad\subset\quad \underbrace{\quad\mathfrak{H}^{}_3\quad}_{=:\, V_{\mathrm{tot}}} \quad \,.
\end{align}
Note that: (i) $\mathfrak{H}^{}_2 = V_g^{\perp_{\vphantom{}_{\mathrm{KG}}}}$, so that the filtration satisfies the orthogonality condition \eqref{2=1perp} of Araki's standard Gupta-Bleuler triplet. (ii) For the refined filtration, we use the auxiliary labels $\mathfrak{H}^{}_{2^\prime}$ and $\mathfrak{H}_{\overline{2}^\prime} := \mathfrak{H}^{}_2$. Using Table \ref{Table 1}, one readily verifies that ${\mathfrak{H}}^{}_1$, $\mathfrak{H}^{}_{2^\prime}$, ${\mathfrak{H}}^{}_2 = {\mathfrak{H}}_{{\overline{2}^\prime}}$, and ${\mathfrak{H}}^{}_3$ are closed $G=\mathrm{SO}_0(1,4)$-invariant subspaces with respect to the KG form, and that each successive inclusion is non-split, in the sense that the preceding invariant subspace admits no closed $G$-invariant complement in the next. Choosing the corresponding complements as in \eqref{R1+...+Rn''}, one obtains
\begin{align} \label{Araki-direct-sum-bulk333}
    V_{\mathrm{tot}} \quad=\quad \underbrace{\qquad{\mathfrak{R}}^{}_1\qquad}_{\simeq\, \mathfrak{H}^{}_1 / \mathfrak{H}^{}_0 \,\simeq\, V_g} \quad\bigoplus\quad \underbrace{\qquad{\mathfrak{R}}^{}_2\qquad}_{\simeq\, {\mathfrak{H}}^{}_2/{\mathfrak{H}}^{}_1 \,\simeq\, V^\prime\oplus \overline{V^\prime}} \;\overset{\scalebox{1.6}[1.3]{$\curvearrowright$}}{\Bigg(=} \; \underbrace{\qquad{\mathfrak{R}}^{}_{2^\prime}\qquad}_{\simeq\, {\mathfrak{H}}_{2^\prime}/{\mathfrak{H}}^{}_1 \,\simeq\, V^\prime} \quad\bigoplus\quad \underbrace{\qquad{\mathfrak{R}}_{{\overline{2}^\prime}}\qquad}_{\simeq\,{\mathfrak{H}}_{{\overline{2}^\prime}}/{\mathfrak{H}}^{}_{2^\prime} \,\simeq\, \overline{V^\prime}} \Bigg) \quad\bigoplus\quad \underbrace{\qquad{\mathfrak{R}}^{}_3\qquad}_{\simeq\, {\mathfrak{H}}^{}_3/{\mathfrak{H}}^{}_2 \,\simeq\, V_s} \;\;\,. 
\end{align}

Henceforth, the symbol ``$\,\overset{\scalebox{1.6}[1.3]{$\curvearrowright$}}{}\,$'' denotes passage from the standard three-layer description to its refined four-layer counterpart (see Fig. \ref{Fig}). Specifically, the object immediately preceding the symbol, associated with the unresolved intermediate sector ${\mathfrak{R}}^{}_2$, is replaced by the corresponding pair obtained upon resolving ${\mathfrak{R}}^{}_2={\mathfrak{R}}^{}_{2^\prime}\oplus{\mathfrak{R}}_{{\overline{2}^\prime}} \simeq V^\prime\oplus\overline{V^\prime}$ into two distinct layers, while the gauge and supplementary sectors ${\mathfrak{R}}^{}_1$ and ${\mathfrak{R}}^{}_3$ remain unchanged. This convention applies uniformly to the corresponding sectors, projections, representations, extension maps, and their actions.

\begin{figure}[t]
    \centering
    \begin{tikzpicture}[
        every node/.style={align=center},
        layer/.style={
            draw=black!75,
            thick,
            rounded corners=2pt,
            minimum width=6.0cm,
            minimum height=1.05cm,
            inner sep=0pt
        },
        >=Stealth
    ]

    \node[layer] (R3) at (0,1.575)
        {$\mathfrak{R}_3 \simeq V_s$};

    \node[layer] (R2b) at (0,0.525)
        {$\mathfrak{R}_{{\overline{2}^\prime}}
        \simeq \overline{V^\prime}$};

    \node[layer] (R2p) at (0,-0.525)
        {$\mathfrak{R}_{2^\prime}
        \simeq V^\prime$};

    \node[layer] (R1) at (0,-1.575)
        {$\mathfrak{R}_1 \simeq V_g$};

    \draw[<->, thick]
        (R3.west)
        .. controls (-4.15,1.575) and (-4.15,-1.575) ..
        (R1.west);

    \node[anchor=east] at (-3.99,0)
        {KG pairing};

    \draw[
        decorate,
        decoration={brace,amplitude=6pt},
        thick
    ]
        (3.25,1.02) -- (3.25,-1.02)
        node[midway,right=10pt]
        {$\mathfrak{R}_2
        =
        \mathfrak{R}_{2^\prime}
        \oplus
        \mathfrak{R}_{{\overline{2}^\prime}}$};

    \node at (0,-2.55)
        {$V_{\mathrm{tot}}
        =
        \mathfrak{R}_1
        \oplus
        \mathfrak{R}_{2^\prime}
        \oplus
        \mathfrak{R}_{{\overline{2}^\prime}}
        \oplus
        \mathfrak{R}_3$};

    \end{tikzpicture}

    \caption{Schematic representation of the refined Araki decomposition for the dS scalar discrete-series realization $\Pi_{p,0}$. The intermediate sector $\mathfrak{R}_2 = V^\prime \oplus \overline{V^\prime}$ of the standard Gupta-Bleuler triplet is resolved into the two irreducible layers $\mathfrak{R}_{2^\prime} = V^\prime$ and $\mathfrak{R}_{{\overline{2}^\prime}} = \overline{V^\prime}$. The outer layers $\mathfrak{R}_1 = V_g$ and $\mathfrak{R}_3 = V_s$ are mutually adjoint-conjugate and non-degenerately paired by the KG form.}
    \label{Fig}
\end{figure}
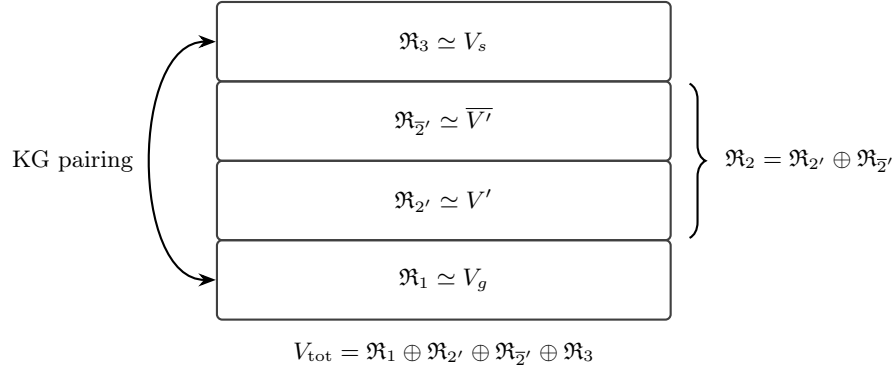

\begin{Theorem}\label{main thoerem}

Let $\pi\colon\; G = \mathrm{SO}_0(1,4)\longrightarrow\mathrm{GL}(V_{\mathrm{tot}})$ denote the representation associated with the indecomposable realization of the scalar discrete-series UIR $\Pi_{p,0}$, for a given $p=1,2,\cdots$. Then, $\pi$ admits a Gupta-Bleuler triplet in the standard form \eqref{standard}, together with a refined four-layer realization obtained by resolving its reducible intermediate sector into its two irreducible constituents. The corresponding three- and four-layer realizations are displayed below, respectively:
\begin{align}
    \pi(\texttt{g}) =
    \begin{pmatrix}
        \pi^{}_1 & c_{12} & c_{13} \\
            0 & \pi_2 & c_{23} \\
            0 & 0 & \pi_1^\sharp
    \end{pmatrix} \,, 
    \qquad\;\overset{\scalebox{2.6}[2.3]{$\curvearrowright$}}{\qquad^{}} \qquad \pi=\;\vcenter{\hbox{%
\begin{tikzpicture}[baseline=(m.center)]
  \matrix (m) [
    matrix of math nodes, 
    left delimiter={(}, 
    right delimiter={)}, 
    nodes in empty cells,
    nodes={inner sep=3pt, anchor=center},
    ampersand replacement=\&,
    column 1/.style={nodes={minimum width=1.0cm}},
    column 2/.style={nodes={minimum width=1.0cm}},
    column 3/.style={nodes={minimum width=2.0cm}},
    column 4/.style={nodes={minimum width=1.4cm}}
  ] {
    \pi^{}_1\& c_{1{2^\prime}} \& c_{1{\overline{2}^\prime}} \& c_{13} \\
    0 \& \pi^{}_{2^\prime} \& c_{{2^\prime}{{\overline{2}^\prime}}}=0 \& c_{{2^\prime}3} \\
    0 \& 0 \& \pi_{{\overline{2}^\prime}} \& c_{{\overline{2}^\prime} 3} \\
    0 \& 0 \& 0 \& \pi_1^\sharp  \\
  };
  \draw[dashed, gray] (m-1-1.north east) -- (m-4-1.south east);
  \draw[dashed, gray] (m-1-3.north east) -- (m-4-3.south east);
  \draw[dashed, gray] (m-1-1.south west) -- (m-1-4.south east);
  \draw[dashed, gray] (m-3-1.south west) -- (m-3-4.south east);
\end{tikzpicture}
}}\,,
\end{align}
\begin{align}
\pi = \pi_1^\sharp \longrightarrow \pi_2 \longrightarrow \pi^{}_1\,,
\qquad\;\overset{\scalebox{2.6}[2.3]{$\curvearrowright$}}{\qquad^{}} \qquad
\pi = \vcenter{\hbox{%
  \xymatrix@C=1.5em@R=1em{
    & \pi^{}_{2^\prime} \ar[dr] & \\
    \pi_1^\sharp \ar[ur] \ar[dr] & & \pi^{}_1\\
    & \pi_{{\overline{2}^\prime}} \ar[ur] &
  }%
}}\,, \hspace{3.1cm}
\end{align}
such that,
\begin{align} \label{yields}
    0 \neq  c_{12} \in H^1(\pi^{}_2, \pi^{}_1) \,, \qquad\;\overset{\scalebox{2.6}[2.3]{$\curvearrowright$}}{\qquad^{}}\qquad 
    \begin{cases}
        0 \neq c_{12^\prime}\in H^1(\pi^{}_{2^\prime},\pi^{}_1) \,, \\[5pt]
        0 \neq c_{1{{\overline{2}^\prime}}}\in H^1(\pi_{{\overline{2}^\prime}},\pi^{}_1) \,,
    \end{cases} 
\end{align}
\begin{align}
    0 \neq  c_{12}^\ast \in H^1(\pi_1^\sharp,\pi^{}_2) \,, \qquad\;\overset{\scalebox{2.6}[2.3]{$\curvearrowright$}}{\qquad^{}}\qquad
    \begin{cases}
        0 \neq c^\ast_{12^\prime} \in H^1(\pi^{\sharp}_1, \pi^{}_{2^\prime}) \,, \\[5pt]
        0 \neq c^\ast_{1{{\overline{2}^\prime}}} \in H^1(\pi^{\sharp}_1, \pi_{{\overline{2}^\prime}}) \,,
    \end{cases} 
\end{align}
\begin{align}
    c^{}_{12} \times c_{12}^\ast \in B_s^2(\pi_1^\sharp, \pi^{}_1) \,, \qquad\;\overset{\scalebox{2.6}[2.3]{$\curvearrowright$}}{\qquad^{}}\qquad c^{}_{12^\prime_{}}\times c_{12^\prime_{}}^{\ast} + c^{}_{1{\overline{2}^\prime}}\times c_{1{\overline{2}^\prime}}^{\ast} \in B_s^2(\pi_1^\sharp ,\, \pi^{}_1) \,.
\end{align}
Here, $\pi^{}_{2^\prime}$ and $\pi_{{\overline{2}^\prime}}$ denote two equivalent realizations of the UIR $\Pi_{p,0}$, carried by $V^\prime$ and $\overline{V^\prime}$, respectively. By contrast, $\pi^{}_1$ and $\pi_1^\sharp$ denote mutually adjoint-conjugate finite-dimensional representations, carried by $V_g$ and $V_s$, respectively, each Weyl equivalent to $\Pi_{p,0}$, in the sense of possessing the same eigenvalues of the dS Casimir operators as $\Pi_{p,0}$.

\end{Theorem}

\begin{proof}

\textbf{Step I.} Let $P_1$, $P_2 \overset{\scalebox{1.6}[1.3]{$\curvearrowright$}}{\big(=}P_{2^\prime} + P_{{\overline{2}^\prime}})$, and $P_3$ denote the projections associated with the decomposition \eqref{Araki-direct-sum-bulk333}. Explicitly (see Table \ref{Table 1}),
\begin{align}
    P_1 &\colon\; V_{\mathrm{tot}} \; \longrightarrow\; V_g \,, \qquad P_1 \Phi = \sum_{L=0}^{p-1} \left\langle \Phi \,,\, \phi^{(s)}_L \right\rangle_{\vphantom{}_{\mathrm{KG}}} \;\phi^{(g)}_L \,, \\[3pt]
    P_2 \overset{\scalebox{1.6}[1.3]{$\curvearrowright$}}{\big(=}P_{2^\prime} + P_{{\overline{2}^\prime}}) &\colon\; V_{\mathrm{tot}} \; \longrightarrow\; V^\prime\oplus \overline{V^\prime} \,, \qquad P_2 \Phi = \underbrace{\sum_{L=p}^\infty \left\langle \Phi \,,\, \phi^{}_L \right\rangle_{\vphantom{}_{\mathrm{KG}}} \;\phi^{}_L}_{P^{}_{2^\prime} \colon\; V_{\mathrm{tot}} \; \longrightarrow\; V^\prime} \;+\;\, \underbrace{\sum_{L=p}^\infty \left\langle \Phi \,,\, -\overline{\phi^{}_L} \right\rangle_{\vphantom{}_{\mathrm{KG}}} \;\overline{\phi^{}_L}}_{P_{{\overline{2}^\prime}} \colon\; V_{\mathrm{tot}} \; \longrightarrow\; \overline{V^\prime}} \,, \\[5pt] 
    P_3 &\colon\; V_{\mathrm{tot}} \; \longrightarrow\; V_s \,, \qquad P_3 \Phi = \sum_{L=0}^{p-1} \left\langle \Phi \,,\, \phi^{(g)}_L \right\rangle_{\vphantom{}_{\mathrm{KG}}} \;\phi^{(s)}_L \,.
\end{align}
We recall that $\{\phi^{(g)}_{L<p}\}$ and $\{\phi^{(s)}_{L<p}\}$ are KG-paired bases of $V_g$ and $V_s$, respectively, while $\{\phi_{L\geqslant p}\}$ and $\{\overline{\phi_{L\geqslant p}}\}$ are KG-``orthonormal'' bases of $V^\prime$ and $\overline{V^\prime}$, respectively; see Table \ref{Table 1}. We note that every vector $\Phi\in V_{\mathrm{tot}}$ admits the unique decomposition
\begin{align}\label{p1+p2+p3}
    \Phi &=\quad P_1 \Phi \quad+\quad P_2 \Phi \;\overset{\scalebox{1.6}[1.3]{$\curvearrowright$}}{\Big(=} \; P_{2^\prime}\Phi + P_{{\overline{2}^\prime}} \Phi\Big)  \quad+\quad P_3 \Phi \,.
\end{align}
Moreover, the maps $P_i$, $i=1,2 \overset{\scalebox{1.6}[1.3]{$\curvearrowright$}}{\big(=}2^\prime,{\overline{2}^\prime}\big), 3$, satisfy
\begin{align}
    P_iP_j=\delta_{ij}P_i \,,
\end{align}
where $\delta_{ij}$ denotes the Kronecker delta.

\textbf{Step II.} Following Araki's notation, let $\pi \colon\; G = \mathrm{SO}_0(1,4) \longrightarrow \mathrm{GL}(V_{\mathrm{tot}})$ denote the representation associated with the indecomposable scalar discrete-series realization $\Pi_{p,0}$, $p=1,2,\cdots$. With respect to the decomposition \eqref{p1+p2+p3} and the corresponding dS-invariant subspace structure summarized in Table \ref{Table 1}, we obtain
\begin{align}
    \label{la-1} \pi(\texttt{g})P_1 &\;=\quad \underbrace{P_1 \pi(\texttt{g})P_1}_{=:\, \pi^{}_1(\texttt{g})} \quad+\quad \underbrace{ P_2 \pi(\texttt{g})P_1}_{=:\, c_{21}(\texttt{g}) \,=\, 0} \;\overset{\scalebox{1.6}[1.3]{$\curvearrowright$}}{\Bigg(=}\;\underbrace{ P_{2^\prime}\pi(\texttt{g})P_1}_{=:\, c_{2^\prime 1}(\texttt{g})\,=\,0} \;+\; \underbrace{ P_{{\overline{2}^\prime}} \pi(\texttt{g})P_1}_{=:\, c_{{\overline{2}^\prime} 1}(\texttt{g})\,=\,0}\Bigg) \quad+\quad \underbrace{P_3 \pi(\texttt{g})P_1}_{=:\, c_{31}(\texttt{g}) \,=\, 0} \nonumber\\[3pt]
    &\;=\quad \pi^{}_1(\texttt{g}) \,,
\end{align}
\begin{align}   
    \label{la-2} \pi(\texttt{g})P_2 \;\overset{\scalebox{1.6}[1.3]{$\curvearrowright$}}{\Bigg(=} 
    \begin{array}{cc}
        &\pi(\texttt{g}) P_{2^\prime} \\[3pt]
        +\!\!\!\!\!\!& \\[3pt]
        &\pi(\texttt{g}) P_{{\overline{2}^\prime}}     
    \end{array}
    \Bigg) 
    &\;=\quad \underbrace{P_1 \pi(\texttt{g})P_2}_{=:\, c_{12}(\texttt{g})} \;\overset{\scalebox{1.6}[1.3]{$\curvearrowright$}}{\Bigg(=}
    \begin{array}{cc}
        &\underbrace{P_1\pi(\texttt{g})P_{2^\prime}}_{=:\, c_{12^\prime}(\texttt{g})} \\[3pt]
        +\!\!\!\!\!\!& \\[3pt]
        &\underbrace{P_1\pi(\texttt{g})P_{{\overline{2}^\prime}}}_{=:\, c_{1{\overline{2}^\prime}}(\texttt{g})}
    \end{array}
    \Bigg) \quad+\quad \underbrace{P_2 \pi(\texttt{g})P_2}_{=:\, \pi^{}_2(\texttt{g})} \;\overset{\scalebox{1.6}[1.3]{$\curvearrowright$}}{\Bigg(=}
    \begin{array}{cc}
        &\underbrace{P_{2^\prime}\pi(\texttt{g})P_{2^\prime}}_{=:\, \pi^{}_{2^\prime}(\texttt{g})} \;+\; \underbrace{ P_{{\overline{2}^\prime}}\pi(\texttt{g})P_{2^\prime}}_{=:\, c_{{\overline{2}^\prime} 2^\prime}(\texttt{g}) \,=\, 0} \\[3pt]
        +\!\!\!\!\!\!& \\[3pt]
        &\underbrace{ P_{2^\prime}\pi(\texttt{g})P_{{\overline{2}^\prime}}}_{=:\, c_{2^\prime{\overline{2}^\prime}}(\texttt{g}) \,=\, 0} \;+\; \underbrace{ P_{{\overline{2}^\prime}}\pi(\texttt{g})P_{{\overline{2}^\prime}}}_{=:\, \pi_{{\overline{2}^\prime}}(\texttt{g})}
    \end{array} 
    \Bigg)\nonumber\\[8pt] 
    & \qquad + \underbrace{P_3 \pi(\texttt{g})P_2}_{=:\, c_{32}(\texttt{g}) \,=\, 0} \;\overset{\scalebox{1.6}[1.3]{$\curvearrowright$}}{\Bigg(=}
    \begin{array}{cc}
        &\underbrace{ P_3\pi(\texttt{g})P_{2^\prime}}_{=:\, c_{32^\prime}(\texttt{g}) \,=\,0} \\[3pt]
        +\!\!\!\!\!\!& \\[3pt]
        &\underbrace{ P_3\pi(\texttt{g})P_{{\overline{2}^\prime}}}_{=:\, c_{3{\overline{2}^\prime}}(\texttt{g})\,=\,0}
    \end{array}
    \Bigg) \nonumber\\[8pt]
    &=\quad c_{12}(\texttt{g}) \;\overset{\scalebox{1.6}[1.3]{$\curvearrowright$}}{\Big(=} 
    \begin{array}{cc}
        &c_{12^\prime}(\texttt{g}) \\[3pt]
        +\!\!\!\!\!\!& \\[3pt]
        &c_{1{\overline{2}^\prime}}(\texttt{g})
    \end{array}
    \Big) \quad+\quad \pi^{}_2(\texttt{g}) \;\overset{\scalebox{1.6}[1.3]{$\curvearrowright$}}{\Big(=}
    \begin{array}{cc}
        &\pi^{}_{2^\prime}(\texttt{g}) \\
        +\!\!\!\!\!\!& \\[3pt]
        &\pi_{{\overline{2}^\prime}}(\texttt{g})
    \end{array} \Big)\,, 
\end{align}
or, more precisely, $\pi(\texttt{g}) P_{2^\prime} = c_{12^\prime}(\texttt{g}) +  \pi^{}_{2^\prime}(\texttt{g})$ and $\pi(\texttt{g}) P_{{\overline{2}^\prime}} = c_{1{\overline{2}^\prime}}(\texttt{g}) +  \pi_{{\overline{2}^\prime}}(\texttt{g})$, and
\begin{align}    
    \label{la-3} \pi(\texttt{g})P_3 &\;=\quad \underbrace{P_1 \pi(\texttt{g})P_3}_{=:\,c_{13}(\texttt{g})} \quad+\quad \underbrace{P_2 \pi(\texttt{g})P_3}_{=:\,c_{23}(\texttt{g})} \;\overset{\scalebox{1.6}[1.3]{$\curvearrowright$}}{\Bigg(=}\;\underbrace{ P_{2^\prime}\pi(\texttt{g})P_3}_{=:\, c_{2^\prime 3}(\texttt{g})} \;+\; \underbrace{ P_{{\overline{2}^\prime}} \pi(\texttt{g})P_3}_{=:\, c_{{\overline{2}^\prime} 3}(\texttt{g})}\Bigg) \quad+\quad \underbrace{P_3 \pi(\texttt{g})P_3}_{=:\,\pi^{}_3(\texttt{g})} \nonumber\\[3pt]
    &\;=\quad c_{13}(\texttt{g}) \quad+\quad c_{23}(\texttt{g}) \;\overset{\scalebox{1.6}[1.3]{$\curvearrowright$}}{\Big(=} c_{2^\prime 3}(\texttt{g}) + c_{{\overline{2}^\prime} 3}(\texttt{g})\Big) \quad+\quad \pi^{}_3(\texttt{g}) \,.
\end{align}
Then, the action of the dS representation $\pi\colon\; G=\mathrm{SO}_0(1,4) \longrightarrow \mathrm{GL}(V_{\mathrm{tot}})$ takes the form\footnote{Note that, by construction, we have 
\begin{align*}
    \pi^{}_1(\texttt{g}) \Phi = \pi^{}_1(\texttt{g}) \phi^{(g)} \,, \quad \pi^{}_2(\texttt{g}) \Phi = \pi^{}_2(\texttt{g}) \phi^{(2)}\,, \quad \pi^{}_3(\texttt{g}) \Phi = \pi^{}_3(\texttt{g}) \phi^{(s)}\,, \quad c_{12}(\texttt{g}) \Phi = c_{12}(\texttt{g}) \phi^{(2)} \,, \quad c_{23}(\texttt{g}) \Phi = c_{23}(\texttt{g}) \phi^{(s)} \,,
\end{align*}
\begin{align*}
    \pi^{}_{2^\prime}(\texttt{g}) \Phi = \pi^{}_{2^\prime}(\texttt{g}) \phi \,, \quad \pi_{{\overline{2}^\prime}}(\texttt{g}) \Phi = \pi_{{\overline{2}^\prime}}(\texttt{g}) \overline{\phi}\,, \quad c_{12^\prime}(\texttt{g}) \Phi = c_{12^\prime}(\texttt{g}) \phi \,, \quad c_{1{\overline{2}^\prime}}(\texttt{g}) \Phi = c_{1{\overline{2}^\prime}}(\texttt{g}) \overline{\phi}\,, \quad c_{2^\prime 3}(\texttt{g}) \Phi = c_{2^\prime 3}(\texttt{g}) \phi^{(s)} \,, \quad c_{{\overline{2}^\prime} 3}(\texttt{g}) \Phi = c_{{\overline{2}^\prime} 3}(\texttt{g}) \phi^{(s)} \,. 
\end{align*}}
\begin{align}\label{27}
    \pi(\texttt{g})\,\Phi &=\quad \pi(\texttt{g}) P_1 \,\Phi \quad+\quad \pi(\texttt{g}) P_2 \,\Phi \;\overset{\scalebox{1.6}[1.3]{$\curvearrowright$}}{\Big(=} \pi(\texttt{g}) P_{2^\prime}\,\Phi + \pi(\texttt{g}) P_{{\overline{2}^\prime}}\,\Phi \Big) \quad+\quad \pi(\texttt{g}) P_3 \,\Phi \nonumber\\[3pt]
    &=:\quad \pi^{}_1(\texttt{g}) \,\phi^{(g)} \nonumber\\
    &\quad +\quad c_{12}(\texttt{g})\phi^{(2)} \;\overset{\scalebox{1.6}[1.3]{$\curvearrowright$}}{\Big(=} c_{12^\prime}(\texttt{g})\phi + c_{1{\overline{2}^\prime}}(\texttt{g})\overline{\phi} \Big) \quad+\quad \pi^{}_2(\texttt{g})\phi^{(2)} \;\overset{\scalebox{1.6}[1.3]{$\curvearrowright$}}{\Big(=} \pi^{}_{2^\prime}(\texttt{g})\phi + \pi_{{\overline{2}^\prime}}(\texttt{g})\overline{\phi} \Big) \nonumber\\
    &\quad+\quad c_{13}(\texttt{g})\phi^{(s)} \quad+\quad c_{23}(\texttt{g})\phi^{(s)} \;\overset{\scalebox{1.6}[1.3]{$\curvearrowright$}}{\Big(=} c_{2^\prime 3}(\texttt{g})\phi^{(s)} + c_{{\overline{2}^\prime} 3}(\texttt{g})\phi^{(s)} \Big) \quad+\quad \pi^{}_3(\texttt{g}) \phi^{(s)} \,.
\end{align}

\textbf{Step III.} We observe the following:
\begin{enumerate}
    \item{The representation $\pi$ of $G=\mathrm{SO}_0(1,4)$ on $V_{\mathrm{tot}}$ is \emph{reducible} (see footnote \ref{footnote indecom}), since both the standard three-layer filtration and its refined four-layer counterpart contain non-trivial proper closed $G$-invariant subspaces. This invariant subspace structure is summarized in Table \ref{Table 1}.}

    \item{The representation $\pi$ of $G=\mathrm{SO}_0(1,4)$ on $V_{\mathrm{tot}}$ is topologically \emph{indecomposable} (see footnote \ref{footnote indecom}), since there exist no non-trivial closed $G$-invariant subspaces ${\mathfrak{E}}_1$ and ${\mathfrak{E}}_2$ satisfying ${\mathfrak{E}}_1\cap{\mathfrak{E}}_2=\{0\}$ and $(({\mathfrak{E}}_1 + {\mathfrak{E}}_2)^\perp)^\perp=V_{\mathrm{tot}}$ (see Table \ref{Table 1}). This indecomposability is reflected in the non-vanishing extension maps $c_{12} \overset{\scalebox{1.6}[1.3]{$\curvearrowright$}}{(=} c_{12^\prime} + c_{1{\overline{2}^\prime}})$ and $c_{23} \overset{\scalebox{1.6}[1.3]{$\curvearrowright$}}{(=} c_{2^\prime 3} + c_{{\overline{2}^\prime} 3})$, whose presence couples the successive layers of the Araki filtration and prevents the representation from splitting into a direct sum of invariant subrepresentations.}

    \item{The middle representation $\pi_2$ is reducible and topologically decomposable, admitting the direct-sum decomposition into the irreducible representations $\pi^{}_{2^\prime}$ and $\pi_{{\overline{2}^\prime}}$ carried by $V^\prime$ and $\overline{V^\prime}$, respectively. Equivalently, $V^\prime \cap \overline{V^\prime} = \{0\}$ and $\left( \left(V^\prime + \overline{V^\prime} \right)^{\perp_{\vphantom{}_{\mathrm{KG}}}} \right)^{\perp_{\vphantom{}_{\mathrm{KG}}}} = V^\prime + \overline{V^\prime}$ (see Table \ref{Table 1} and footnote \ref{footnote indecom}).}

    \item{The induced representations $\pi^{}_1$ on $V_g$ and $\pi^{}_3$ on $V_s$ are irreducible and, consequently, topologically indecomposable.}

    \item{The induced representations $\pi_1$ and $\pi_3$ are mutually adjoint-conjugate,
    \begin{align}\label{3sharp=1}
        \pi^{}_3 \simeq \pi_1^\sharp \,, \quad \pi_3^\sharp \simeq \pi^{}_1 \,,
    \end{align}
    as follows:
    \begin{enumerate}
        \item{From ${V_g}^{\perp_{\vphantom{}_{\mathrm{KG}}}} = V^\prime \oplus \overline{V^\prime} \oplus V_g$ and ${V_s}^{\perp_{\vphantom{}_{\mathrm{KG}}}} = V^\prime \oplus \overline{V^\prime} \oplus V_s$ (see Table \ref{Table 1}), together with the non-degenerate KG pairing between $V_g$ and $V_s$ in Eq. \eqref{gs}, the two sectors are mutually adjoint-conjugate,
        \begin{align}\label{gssharp}
            V_g^\sharp \simeq V_s \,, \quad V_s^\sharp \simeq V_g \,.
        \end{align}}

        \item{Using Eqs. \eqref{la-1} and \eqref{la-3}, the KG-orthogonality relations summarized in Table \ref{Table 1}, and the $G$-invariance of the KG form, we obtain
        \begin{align} \label{suprise}
            \left\langle \pi_3(\texttt{g})\phi^{(s)} \,,\, \pi_1(\texttt{g})\phi^{(g)} \right\rangle_{\vphantom{}_{\mathrm{KG}}} =&\, \left\langle \pi(\texttt{g})\phi^{(s)} \,,\, \pi(\texttt{g})\phi^{(g)} \right\rangle_{\vphantom{}_{\mathrm{KG}}} \nonumber\\[5pt]
            &\, -\; \underbrace{\left\langle c_{23}(\texttt{g})\phi^{(s)} \,,\, \pi(\texttt{g})\phi^{(g)} \right\rangle_{\vphantom{}_{\mathrm{KG}}}}_{=\,0} -\; \underbrace{\left\langle c_{13}(\texttt{g})\phi^{(s)} \,,\, \pi(\texttt{g})\phi^{(g)} \right\rangle_{\vphantom{}_{\mathrm{KG}}}}_{=\,0} \nonumber\\
            =&\, \left\langle \phi^{(s)} \,,\, \phi^{(g)} \right\rangle_{\vphantom{}_{\mathrm{KG}}} \,.
        \end{align}
        Hence, by the definition \eqref{conj-def} and Eq. \eqref{gssharp}, $\pi^{}_3 \simeq \pi_1^\sharp$ and, equivalently, $\pi_3^\sharp \simeq \pi^{}_1$.}
    \end{enumerate}}

    \item{By the same argument as in Eq. \eqref{suprise}, applied to the non-degenerate restriction of the KG form to $V^\prime \oplus \overline{V^\prime}$ and to each of its two components, one likewise obtains
    \begin{align} \label{sharps}
        \pi^{}_2 \simeq \pi_2^\sharp \,, \quad \pi^{}_{2^\prime} \simeq \pi^\sharp_{2^\prime} \,, \quad \pi_{{\overline{2}^\prime}} \simeq \pi^\sharp_{{\overline{2}^\prime}} \,.
    \end{align}}
    
    \item{The cochain $c_{23}$ is the adjoint of $c_{12}$ in the sense of Eq. \eqref{eq:3.1},
    \begin{align}\label{akse0}
        c^{}_{23} = c_{12}^{\ast} \,.
    \end{align}
    Indeed, using Eqs. \eqref{la-3} and \eqref{la-2}, the KG-orthogonality relations summarized in Table \ref{Table 1}, and the $\mathrm{SO}_0(1,4)$-invariance of the KG form, we obtain 
    \begin{align}\label{Alan-1}
        \left\langle c_{23}(\texttt{g}^{-1})\phi^{(s)} \;,\; \phi_2 \right\rangle_{\vphantom{}_\mathrm{KG}} &= \left\langle \pi(\texttt{g}^{-1})\phi^{(s)} \;,\; \phi_2 \right\rangle_{\vphantom{}_\mathrm{KG}} - \;\;\underbrace{\left\langle \pi^{}_3(\texttt{g}^{-1})\phi^{(s)} \;,\; \phi_2 \right\rangle_{\vphantom{}_\mathrm{KG}}}_{=\, 0} - \;\;\underbrace{\left\langle c^{}_{13}(\texttt{g}^{-1})\phi^{(s)} \;,\; \phi_2 \right\rangle_{\vphantom{}_\mathrm{KG}}}_{=\, 0} \nonumber\\
        &= \left\langle \phi^{(s)} \;,\; \pi(\texttt{g})\phi_2 \right\rangle_{\vphantom{}_\mathrm{KG}} \nonumber\\
        &= \underbrace{\left\langle \phi^{(s)} \;,\; \pi^{}_2(\texttt{g})\phi_2 \right\rangle_{\vphantom{}_\mathrm{KG}}}_{=\, 0} + \;\;\left\langle \phi^{(s)} \;,\; c_{12}(\texttt{g})\phi_2 \right\rangle_{\vphantom{}_\mathrm{KG}} = \left\langle \phi^{(s)} \;,\; c_{12}(\texttt{g})\phi_2 \right\rangle_{\vphantom{}_\mathrm{KG}} \,.
    \end{align}
    An analogous argument yields
    \begin{align}\label{akse}
        c^{}_{2^\prime 3} = c_{12^\prime}^\ast\,, \quad\mbox{and}\quad c^{}_{{\overline{2}^\prime}3} = c_{1{\overline{2}^\prime}}^\ast \,.
    \end{align}}
\end{enumerate}

\textbf{Step IV.} The action \eqref{27} admits the following upper-triangular matrix realization (left), analogous to the standard form \eqref{standard}. Decomposing the middle sector into its two irreducible components further refines this three-layer structure into the corresponding four-layer form shown on the right,
\begin{align}\label{Araki-matrix-bulk}
    \pi(\texttt{g})\Phi =
    \begin{pmatrix}
        \pi^{}_1(\texttt{g}) & c^{}_{12}(\texttt{g}) & c^{}_{13} (\texttt{g}) \\
        0 & \pi^{}_2(\texttt{g}) & c^{}_{23}(\texttt{g}) \\
        0 & 0 & \pi^{}_3(\texttt{g})
    \end{pmatrix}
    \!\!\begin{pmatrix}
        \phi^{(g)}\\
        \phi^{(2)} \\
        \phi^{(s)}
    \end{pmatrix} \,,
    \quad\;\overset{\scalebox{1.6}[1.3]{$\curvearrowright$}}{\qquad^{}}\quad 
    \pi(\texttt{g})\Phi = \begin{pmatrix}
        \pi^{}_1(\texttt{g}) & c^{}_{1{2^\prime}}(\texttt{g}) & c_{1{{\overline{2}^\prime}}}(\texttt{g}) & c^{}_{13}(\texttt{g}) \\[3pt]
        0 & \pi^{}_{2^\prime}(\texttt{g}) & c^{}_{{2^\prime}{{\overline{2}^\prime}}}(\texttt{g})=0 & c^{}_{{2^\prime}3}(\texttt{g}) \\[3pt]
        0 & 0 & \pi_{{\overline{2}^\prime}}(\texttt{g}) & c^{}_{{{\overline{2}^\prime}}3}(\texttt{g}) \\[3pt]
        0 & 0 & 0 & \pi^{}_3(\texttt{g})
    \end{pmatrix}
    \!\!\begin{pmatrix}
        \phi^{(g)}\\
        \phi \\
        \overline{\phi} \\
        \phi^{(s)}
    \end{pmatrix}\,.
\end{align}

By construction, the block operators $\pi_i$ and $c_{ij}$, with $i,j=1,2\overset{\scalebox{1.6}[1.3]{$\curvearrowright$}}{\big(=}2^\prime,{\overline{2}^\prime}\big),3$, satisfy the following fundamental relations, obtained directly from their definitions in terms of the projections $P_i$ (see Eqs. \eqref{la-1}-\eqref{la-3}) by applying the representation law $\pi(\texttt{g}_1\texttt{g}_2)=\pi(\texttt{g}_1)\pi(\texttt{g}_2)$, for all $\texttt{g}_1,\texttt{g}_2 \in G = \mathrm{SO}_0(1,4)$.\footnote{As a representative verification, we establish Eq. \eqref{5-re} directly from the definitions of the corresponding block operators and the representation law; the remaining identities follow analogously. Starting from $\pi(\texttt{g}_1\texttt{g}_2)=\pi(\texttt{g}_1)\pi(\texttt{g}_2)$ and projecting onto the $(1,3)$ sector by acting with $P_1$ from the left and $P_3$ from the right, we obtain
\begin{align*} \label{cg2-composition}
    P_1 \pi(\texttt{g}_1\texttt{g}_2) P_3 = P_1\pi(\texttt{g}_1) \pi(\texttt{g}_2) P_3&\,, \\
    \quad\Longrightarrow\quad c^{}_{13} (\texttt{g}_1\texttt{g}_2) =&\, P_1 \Big( \pi_1(\texttt{g}_1) + c_{12}(\texttt{g}_1) + c_{13}(\texttt{g}_1) + \pi_2(\texttt{g}_1) + c_{23}(\texttt{g}_1) + \pi_3(\texttt{g}_1) \Big) \\
    &\times \Big( \pi_1(\texttt{g}_2) + c_{12}(\texttt{g}_2) + c_{13}(\texttt{g}_2) + \pi_2(\texttt{g}_2) + c_{23}(\texttt{g}_2) + \pi_3(\texttt{g}_2) \Big)P_3 \tag{by Eq. \eqref{27}} \\
    =&\, P_1 \Big( \pi_1(\texttt{g}_1) + c_{12}(\texttt{g}_1) + c_{13}(\texttt{g}_1) \Big) \Big( c_{13}(\texttt{g}_2) + c_{23}(\texttt{g}_2) + \pi_3(\texttt{g}_2) \Big) P_3 \\
    =&\, \pi_1(\texttt{g}_1) c_{13}(\texttt{g}_2) + c_{12}(\texttt{g}_1) c_{23}(\texttt{g}_2) + c_{13}(\texttt{g}_1) \pi_3(\texttt{g}_2) \,.
\end{align*}} Equivalently, the same relations are independently recovered by multiplying the upper-triangular matrix realizations \eqref{Araki-matrix-bulk}. Namely,
\begin{align}\label{1-re}
    \pi^{}_1(\texttt{g}_1\texttt{g}_2) = \pi^{}_1(\texttt{g}_1)\pi^{}_1(\texttt{g}_2) \,,
\end{align}
\begin{align}
    \label{2-re} \pi^{}_2(\texttt{g}_1\texttt{g}_2) = \pi^{}_2(\texttt{g}_1)\pi^{}_2(\texttt{g}_2) \,, \quad\overset{\scalebox{1.6}[1.3]{$\curvearrowright$}}{\qquad^{}}\quad 
    \begin{cases}
        \pi^{}_{2^\prime}(\texttt{g}_1\texttt{g}_2)
        = \pi^{}_{2^\prime}(\texttt{g}_1)\pi^{}_{2^\prime}(\texttt{g}_2) \,, \\[5pt]
        \pi_{{\overline{2}^\prime}}(\texttt{g}_1\texttt{g}_2)
        = \pi_{{\overline{2}^\prime}}(\texttt{g}_1)\pi_{{\overline{2}^\prime}}(\texttt{g}_2) \,,
    \end{cases} 
\end{align}
\begin{align} \label{3-re} 
    \pi^{}_3(\texttt{g}_1\texttt{g}_2) = \pi^{}_3(\texttt{g}_1)\pi^{}_3(\texttt{g}_2) \,, 
\end{align}
\begin{align} \label{4-re}
    \begin{cases}
        c_{12}(\texttt{g}_1\texttt{g}_2) = \pi^{}_1(\texttt{g}_1)c_{12}(\texttt{g}_2) + c_{12}(\texttt{g}_1)\pi^{}_2(\texttt{g}_2) \,, \\[5pt]
        \quad \Longrightarrow\quad (\delta c_{12})(\texttt{g}_1,\texttt{g}_2) = 0 \,, \qquad \text{(by Eq. \eqref{sum-cn})} \\[5pt]
        \quad \Longrightarrow\quad 0\neq \! \underbrace{c_{12} \in H^1(\pi^{}_2, \pi^{}_1)}_{\mbox{see note (1) below}} \,,
    \end{cases}   
    \quad\overset{\scalebox{1.6}[1.3]{$\curvearrowright$}}{\qquad^{}}\quad 
    \begin{cases}
        c_{12^\prime}(\texttt{g}_1\texttt{g}_2) = \pi^{}_1(\texttt{g}_1)c_{1{2^\prime}}(\texttt{g}_2) + c_{1{2^\prime}}(\texttt{g}_1)\pi^{}_{2^\prime}(\texttt{g}_2) \,, \\[5pt]
        \quad \Longrightarrow\quad (\delta c_{12^\prime})(\texttt{g}_1, \texttt{g}_2) = 0 \,, \qquad \text{(by Eq. \eqref{sum-cn})} \\[5pt]
        \quad \Longrightarrow\quad 0 \neq \underbrace{c^{}_{12^\prime} \in H^1(\pi^{}_{2^\prime},\pi^{}_1)}_{\mbox{see note (1) below}} \,, \\[5pt]
        \\
        c_{1{\overline{2}^\prime}}(\texttt{g}_1\texttt{g}_2) = \pi^{}_1(\texttt{g}_1)c_{1{\overline{2}^\prime} }(\texttt{g}_2) + c_{1{\overline{2}^\prime}}(\texttt{g}_1)\pi_{{\overline{2}^\prime}}(\texttt{g}_2) \,, \\[5pt]
        \quad \Longrightarrow\quad (\delta c_{1{\overline{2}^\prime}})(\texttt{g}_1,\texttt{g}_2) = 0 \,, \qquad \text{(by Eq. \eqref{sum-cn})} \\[5pt]
        \quad \Longrightarrow\quad 0 \neq \underbrace{c^{}_{1{\overline{2}^\prime}} \in H^1(\pi_{{\overline{2}^\prime}},\pi^{}_1)}_{\mbox{see note (1) below}} \,,  
    \end{cases} 
\end{align}
\begin{align} \label{5-re} 
    \begin{cases}
        c_{13}(\texttt{g}_1\texttt{g}_2) =\pi^{}_1(\texttt{g}_1)c_{13}(\texttt{g}_2) + c_{12}(\texttt{g}_1)c_{23}(\texttt{g}_2) \;\overset{\scalebox{1.6}[1.3]{$\curvearrowright$}}{\big(=}c^{}_{12^\prime}(\texttt{g}_1) c^{}_{2^\prime3}(\texttt{g}_2) + c^{}_{1{\overline{2}^\prime}}(\texttt{g}_1) c^{}_{{\overline{2}^\prime} 3}(\texttt{g}_2) \big) + c_{13}(\texttt{g}_1)\pi^{}_3(\texttt{g}_2) \,, \\[5pt]
        \quad \Longrightarrow\quad (\delta c_{13})(\texttt{g}_1,\texttt{g}_2) = - c^{}_{12}(\texttt{g}_1) c^{}_{23}(\texttt{g}_2) \;\overset{\scalebox{1.6}[1.3]{$\curvearrowright$}}{\big(=}- c^{}_{12^\prime}(\texttt{g}_1) c^{}_{2^\prime3}(\texttt{g}_2) - c^{}_{1{\overline{2}^\prime}}(\texttt{g}_1) c^{}_{{\overline{2}^\prime} 3}(\texttt{g}_2) \big) \qquad \text{(by Eq. \eqref{sum-cn})} \\[5pt]
        \hspace{3.5cm} = - c^{}_{12}(\texttt{g}_1) c^\ast_{12}(\texttt{g}_2) \;\overset{\scalebox{1.6}[1.3]{$\curvearrowright$}}{\big(=} - \,c^{}_{12^\prime}(\texttt{g}_1) c^{\ast}_{12^\prime}(\texttt{g}_2) - c^{}_{1{\overline{2}^\prime}}(\texttt{g}_1) c^{\ast}_{1{\overline{2}^\prime}}(\texttt{g}_2)\big) \,, \qquad \text{(by Eqs. \eqref{akse0}, \eqref{akse})} \\[5pt]
        \quad \Longrightarrow\quad c^{}_{12}\times c^\ast_{12} \;\overset{\scalebox{1.6}[1.3]{$\curvearrowright$}}{\big(=} \,c^{}_{12^\prime} \times c^{\ast}_{12^\prime} + c^{}_{1{\overline{2}^\prime}}\times c^{\ast}_{1{\overline{2}^\prime}}\big) \;\;\in\;\; B^2_s(\pi^{}_3\simeq \pi^\sharp_1, \pi^{}_1) \,, \qquad \text{(by Eqs. \eqref{pro-eq}, \eqref{Def-2-co}, \eqref{khodshifte}, \eqref{3sharp=1})}
    \end{cases}
\end{align}
\begin{align} \label{6-re} 
    \begin{cases}
        c^{}_{23}(\texttt{g}_1\texttt{g}_2) = \pi^{}_2(\texttt{g}_1)c^{}_{23}(\texttt{g}_2) + c^{}_{23}(\texttt{g}_1)\pi^{}_3(\texttt{g}_2) \,, \\[5pt]
        \quad \Longrightarrow\quad (\delta c^{}_{23})(\texttt{g}_1,\texttt{g}_2) = 0 \,, \qquad \text{(by Eq. \eqref{sum-cn})} \\[5pt]
        \quad \Longrightarrow\quad \underbrace{c^{}_{23}=c_{12}^\ast\in H^1(\pi^{}_3\simeq \pi_1^\sharp, \pi^{}_2)}_{\mbox{see Eqs. \eqref{akse0}, \eqref{3sharp=1}, note (2) below}} \,,
    \end{cases}    
    \quad\overset{\scalebox{1.6}[1.3]{$\curvearrowright$}}{\qquad^{}}\quad  
    \begin{cases}
        c^{}_{2^\prime 3}(\texttt{g}_1\texttt{g}_2) = \pi^{}_{2^\prime}(\texttt{g}_1)c^{}_{{2^\prime}3}(\texttt{g}_2) + c^{}_{{2^\prime}3}(\texttt{g}_1)\pi^{}_3(\texttt{g}_2) \,, \\[5pt]
        \quad \Longrightarrow\quad (\delta c^{}_{2^\prime3})(\texttt{g}_1,\texttt{g}_2) = 0 \,, \qquad \text{(by Eq. \eqref{sum-cn})} \\[5pt]
        \quad \Longrightarrow\quad 0 \neq \!\!\!\!\! \underbrace{c^{}_{2^\prime 3} = c^{\ast}_{12^\prime} \in H^1(\pi^{}_3 \simeq \pi^\sharp_1, \pi^{}_{2^\prime})}_{\mbox{see Eqs. \eqref{akse}, \eqref{3sharp=1}, note (2) below}} \,, \\[5pt]
        \\
        c^{}_{2^\prime {\overline{2}^\prime}}(\texttt{g}_1\texttt{g}_2) = 0\,, \\[5pt]
        \\
        c^{}_{{\overline{2}^\prime} 3}(\texttt{g}_1\texttt{g}_2) = \pi_{{\overline{2}^\prime}}(\texttt{g}_1)c^{}_{{\overline{2}^\prime} 3}(\texttt{g}_2) + c^{}_{{\overline{2}^\prime} 3}(\texttt{g}_1)\pi^{}_3(\texttt{g}_2) \,, \\[5pt]
        \quad \Longrightarrow\quad (\delta c^{}_{{\overline{2}^\prime}3})(\texttt{g}_1,\texttt{g}_2) = 0 \,, \qquad \text{(by Eq. \eqref{sum-cn})} \\[5pt]
        \quad \Longrightarrow\quad 0 \neq \!\!\!\!\! \underbrace{ c^{}_{{\overline{2}^\prime} 3} = c^{\ast}_{1 {\overline{2}^\prime}} \in H^1(\pi^{}_3 \simeq \pi_1^\sharp, \pi_{{\overline{2}^\prime}})}_{\mbox{see Eqs. \eqref{akse}, \eqref{3sharp=1}, note (2) below}} \,.
    \end{cases} 
\end{align}
Note that a non-trivial class $c_{ij} \in H^1(\pi_j,\pi_i) = Z^1(\pi_j,\pi_i)/B^1(\pi_j,\pi_i)$ is represented by a $1$-cocycle $c_{ij} \in Z^1(\pi_j,\pi_i)$, with $(\delta c_{ij})(\texttt{g}_1,\texttt{g}_2) = 0$, such that $c_{ij} \notin B^1(\pi_j,\pi_i)$, i.e., there exists no fixed $b_{ij} \in C^0(\pi_j,\pi_i)$ satisfying 
\begin{align}
    c_{ij}(\texttt{g}_1) = (\delta b_{ij})(\texttt{g}_1) = \pi_i(\texttt{g}_1)b_{ij} - b_{ij}\pi_j(\texttt{g}_1) \,,
\end{align}
for all $\texttt{g}_1, \texttt{g}_2 \in G = \mathrm{SO}_0(1,4)$. For each of the $1$-cocycles listed above, including $c_{12}$ and $c_{12^\prime}$, the condition $(\delta c_{ij})(\texttt{g}_1, \texttt{g}_2) = 0$ has already been established explicitly; hence $c_{ij}\in Z^1(\pi_j, \pi_i)$. It therefore remains only to show that $c_{ij}\notin B^1(\pi_j,\pi_i)$ in order to conclude that $c_{ij}$ defines a non-trivial class in $H^1(\pi_j,\pi_i)$. Explicitly,
\begin{enumerate}
    \item{$c_{12}\notin B^1(\pi^{}_2,\pi^{}_1)$. If $c_{12}$ were a $1$-coboundary, then, by definition, there would exist a fixed linear map $b_{12} \; \left(\in C^0(\pi^{}_2,\pi^{}_1)\right) \colon\; V^\prime\oplus\overline{V^\prime}\longrightarrow V_g$, such that, for every $\texttt{g}_1 \in G=\mathrm{SO}_0(1,4)$,
    \begin{align}\label{c12} 
        c_{12}(\texttt{g}_1) = (\delta b_{12})(\texttt{g}_1) = \pi^{}_1(\texttt{g}_1)b_{12} - b_{12}\pi^{}_2(\texttt{g}_1) \,.
    \end{align}
    On the other hand, for each $\phi^{(2)} \in V^\prime \oplus \overline{V^\prime}$, Eq. \eqref{la-2} gives
    \begin{align} \label{la-2'}
        \pi(\texttt{g}_1)\phi^{(2)} = \underbrace{\pi^{}_2(\texttt{g}_1) \phi^{(2)}}_{\in\,V^\prime\oplus\overline{V^\prime}} + \,\underbrace{c_{12}(\texttt{g}_1) \phi^{(2)}}_{\in\,V_g} \,.
    \end{align}
    Substituting Eq. \eqref{c12} into \eqref{la-2'}, and using the fact that $\pi(\texttt{g}_1) b_{12} = \pi^{}_1(\texttt{g}_1) b_{12}$ (which follows from $b_{12}\phi^{(2)} \in V_g$ and the $G= \mathrm{SO}_0(1,4)$-invariance of $V_g$), we obtain
    \begin{align} \label{61}
        \pi(\texttt{g}_1)(1 - b_{12}) \phi^{(2)} = (1 - b_{12}) \pi^{}_2(\texttt{g}_1) \phi^{(2)} \,.
    \end{align}
    Notably, the vectors $w := (1-b_{12})\phi^{(2)}$ span a subspace $W \subset \mathfrak{H}_2 = V_g \oplus V^\prime \oplus \overline{V^\prime}$ (see Eq. \eqref{32}). One has $W \cap V_g = \{0\}$: indeed, if $w \in V_g$, then $P_2w = \phi^{(2)} = 0$, and hence $w=0$. Moreover, every $v_g+\phi^{(2)}\in\mathfrak{H}_2$ admits the decomposition
    \begin{align}
        v_g + \phi^{(2)} = \big(v_g + b_{12} \phi^{(2)}\big) + (1 - b_{12}) \phi^{(2)} \,,
    \end{align}
    with the first term belonging to $V_g$ and the second to $W$; thus
    \begin{align}
        \mathfrak{H}_2 = V_g \oplus W \,.
    \end{align}
    Equation \eqref{61} further shows that $W$ is $G=\mathrm{SO}_0(1,4)$-invariant. This contradicts the invariant-subspace structure summarized in Table \ref{Table 1}, according to which the action on $V^\prime\oplus\overline{V^\prime}$ necessarily extends into the gauge sector, so that $V_g$ admits no $G$-invariant complement in $\mathfrak{H}_2$. Hence, no such fixed map $b_{12}$ exists, and therefore $c_{12}\notin B^1(\pi_2,\pi_1)$. Proceeding analogously, one obtains
    \begin{enumerate}
        \item{$c_{12^\prime}\notin B^1(\pi^{}_{2^\prime},\pi^{}_1)$, and}

        \item{$c_{1{{\overline{2}^\prime}}}\notin B^1(\pi_{{\overline{2}^\prime}},\pi^{}_1)$.}
    \end{enumerate}}

    \item{$c^{}_{23}\notin B^1(\pi^{}_3,\pi^{}_2)$. Indeed, suppose, to the contrary, that $c^{}_{23}$ were a $1$-coboundary. Then there would exist a fixed linear map $b^{}_{23}\in C^0(\pi^{}_3,\pi^{}_2)$ such that $c^{}_{23} = (\delta b^{}_{23})$. Taking the adjoint and using $c^{}_{23} = c_{12}^{\ast}$ \eqref{akse0}, together with Eq. \eqref{maryam}, we obtain
    \begin{align}
        c_{12} = c_{23}^{\ast} = (\delta b^{}_{23})^{\ast} = -\delta(b_{23}^{\ast}) \,,
    \end{align}
    where $b_{23}^{\ast} \in C^0(\pi_2^\sharp,\pi_3^\sharp)$, $\pi^{}_2 \simeq \pi_2^\sharp$ \eqref{sharps}, and $\pi_3^\sharp \simeq \pi^{}_1$ \eqref{3sharp=1}. It then follows that $\delta(b_{23}^{\ast})\in B^1(\pi^{}_2,\pi^{}_1)$, and therefore $c_{12}\in B^1(\pi^{}_2,\pi^{}_1)$, contradicting the preceding item. Consequently, $c_{23}\notin B^1(\pi^{}_3,\pi^{}_2)$. Proceeding analogously, one obtains
    \begin{enumerate}
        \item{$c_{2^\prime 3}\notin B^1(\pi^{}_3,\pi^{}_{2^\prime})$, and}

        \item{$c_{{{\overline{2}^\prime}} 3}\notin B^1(\pi^{}_3, \pi_{{\overline{2}^\prime}})$.}
    \end{enumerate}}
\end{enumerate}

\end{proof}

\section{Conclusion}\label{Sec conclu}

In this work, we have uncovered the Araki-cohomological structure underlying the dS-covariant indecomposable Krein realization of the scalar discrete-series representations $\Pi_{p,0}$, $p=1,2,\cdots$, of $\mathrm{SO}_0(1,4)$. We showed that the Krein-Gupta-Bleuler construction realizes Araki's standard Gupta-Bleuler triplet; the gauge and supplementary sectors form an adjoint-conjugate pair, while the successive non-split extensions are represented by non-trivial first cohomology classes related by adjunction. Their product is a self-adjoint $2$-coboundary, thereby realizing the cohomological conditions of Araki's existence theorem explicitly.

Beyond the standard triplet, we showed that the reducibility of the intermediate representation carries additional cohomological structure. Resolving this sector into its two irreducible components decomposes the corresponding extension classes into distinct non-trivial first cohomology classes and yields a four-layer refinement of the Araki filtration. The resulting Gupta-Bleuler quadruplet separates the physical representation from its complex-conjugate sector while preserving the adjoint-paired extension structure linking them to the gauge and supplementary layers.

The scalar dS discrete series thus provides a concrete setting in which the Gupta-Bleuler structure, indecomposability, and group cohomology become different facets of the same representation-theoretic structure. In particular, the passage from the standard triplet to the refined quadruplet shows that resolving a reducible sector can reveal extension data invisible at the coarser level. This suggests a natural route toward the cohomological analysis of more general discrete-series representations \cite{GravGB, takVec, GravPartially} of $\mathrm{SO}_0(1,4)$ and their associated gauge structures.

Another complementary, and perhaps deeper, direction is to determine the role of Araki cohomology in a full quantum field-theoretic realization of the dS discrete series. For the lowest member of the scalar family, $\Pi_{1,0}$, a dS-covariant quantum field theory within the Krein-Gupta-Bleuler framework has already been constructed in Ref. \cite{Takook}. The cohomological structure established here suggests a systematic route toward its generalization to the entire scalar discrete series $\Pi_{p,0}$, $p=1,2,\cdots$. The central problem is to understand how the non-trivial extension classes governing the one-particle indecomposable structure are promoted to the quantum field-theoretic level and how they enter the construction of the field, its state space, and the corresponding physical quotient. In this sense, the Araki cohomology uncovered here is not merely an additional characterization of the one-particle representation, but provides a representation-theoretic framework within which a dS-covariant Krein-Gupta-Bleuler quantum field theory for the full scalar discrete series may be developed. Establishing this connection, and ultimately determining its scope across the dS discrete series, remains an important open problem.

\section*{Acknowledgement}

Maryam Bajalan and Hamed Pejhan are supported by the National Science Fund, Ministry of Education and Science of Bulgaria, under contract KP-06-N92/2.



\begin{thebibliography}{a}

\bibitem{Gazeau2024} M. Enayati, J.-P. Gazeau, H. Pejhan, and A. Wang, \emph{The de Sitter (dS) group and its representations} (2nd edition), Springer, Cham, Switzerland (2024).

\bibitem{Gazeau2026} J.-P. Gazeau, H. Pejhan, and I. Todorov, \emph{Massless Representations in Conformal Space and Their dS Restrictions}, forthcoming in Cambridge Monographs on Mathematical Physics, Cambridge University Press; arXiv:2601.18433 (2026).

\bibitem{Dixmier} J. Dixmier, \emph{Représentations intégrables du groupe de De Sitter}, Bull. Soc. Math. Fr., 89, 9 (1961).

\bibitem{Takahashi} R. Takahashi, \emph{Sur les représentations unitaires des groupes de Lorentz généralisés}, Bull. Soc. Math. Fr., 91, 289 (1963).

\bibitem{GravGB} H. Pejhan, M. Enayati, J.-P. Gazeau, A. Wang, \emph{Gupta-Bleuler quantization for linearized gravity in de Sitter spacetime}, Phys. Rev. D, 100, 066012 (2019).

\bibitem{takVec} T. Garidi, J.-P. Gazeau, S. Rouhani, M.V. Takook, \emph{``Massless'' vector field in de Sitter universe}, J. Math. Phys., 49, 032501 (2008).

\bibitem{GravPartially} J.-P. Gazeau, H. Pejhan, \emph{Covariant quantization of the partially massless graviton field in de Sitter spacetime}, Phys. Rev. D, 108, 065012 (2023).

\bibitem{Bajalan2026} M. Bajalan, J.-P. Gazeau, and H. Pejhan, \emph{dS Scalar Discrete Series: Gupta-Bleuler Structure and Holography}, ArXiv: 2607.17124 (2026), accepted for publication in \emph{Phys. Rev. D}.

\bibitem{Araki-1985-97} H. Araki, \emph{Indecomposable representations with invariant inner product}, Commun. Math. Phys., 97, 149-159 (1985).

\bibitem{Pierotti-1990-31} D. Pierotti, \emph{Some remarks on the Gupta-Bleuler triplet}, J. Math. Phys., 31, 1862-1871 (1990).

\bibitem{Simon-1980} J. Simon, \emph{Introduction to the 1-cohomology of Lie groups}, in \emph{Harmonic Analysis and Representations of Semisimple Lie Groups}, NATO Advanced Study Institutes Series, Vol. 39, edited by P. Trombi (D. Reidel, Dordrecht, 1980), pp. 449-465.

\bibitem{Birrell} N.D. Birrell and P.C.W. Davies, \emph{Quantum fields in curved space}, Cambridge University Press (1982).

\bibitem{Higuchi} A. Higuchi, \emph{Symmetric tensor spherical harmonics on the N-sphere and their application to the de Sitter group $\mathrm{SO}(N, 1)$}, J. Math. Phys., 28(7), 1553-1566 (1987).

\bibitem{Takook} J.-P. Gazeau, J. Renaud, and M.V. Takook, \emph{Gupta-Bleuler quantization for minimally coupled scalar fields in de Sitter space}, Class. Quant. Grav., 17, 1415-1434 (2000).

\end{thebibliography}
\end{document}